\documentclass[10pt, conference, letterpaper]{IEEEtran}
\IEEEoverridecommandlockouts

\usepackage{amsmath}
\usepackage{amsthm}

\usepackage{xcolor}     
\usepackage{graphicx}

\usepackage{algorithm}          
\usepackage{algpseudocode}      

\usepackage{booktabs}
\usepackage{threeparttable}
\usepackage{multirow}
\usepackage{makecell}
\usepackage{adjustbox}
\usepackage{boxedminipage}

\usepackage[utf8]{inputenc}   
\usepackage[english]{babel}
\usepackage{xspace}
\usepackage{url}
\usepackage{comment}
\usepackage{paralist}
\usepackage{wrapfig}
\usepackage{pifont}
\usepackage{cite}

\usepackage{subcaption}

\usepackage[cachedir=.]{minted}
\usepackage{ulem}
\definecolor{blue}{RGB}{10, 10, 200}



\algdef{SE}[UPON]{Upon}{EndUpon}[1]{\textbf{upon}~#1}{} 
\algtext*{EndUpon} 

\newcommand{\bheading}[1]{{\vspace{4pt}\noindent{\textbf{#1}}}}
\newcommand{\iheading}[1]{{\vspace{2pt}\noindent{\textit{#1}}}}
\newcolumntype{?}{!{\vrule width 1pt}}

\newcounter{note}[section]

\newcommand{\secref}[1]{\mbox{Sec.~\ref{#1}}\xspace}

\newcommand{\figref}[1]{\mbox{Fig.~\ref{#1}}}

\newcommand{\ignore}[1]{}

\newcommand{\ie}{\textit{i.e.}\xspace}
\newcommand{\eg}{\textit{e.g.}\xspace}

\newcommand{\etal}{\textit{et al.}\xspace}

\newcommand{\sysname}{\textsc{Hydra}\xspace}

\newcounter{packednmbr}

\newenvironment{packeditemize}{
\begin{list}{$\bullet$}{
\setlength{\labelwidth}{0pt}
\setlength{\itemsep}{2pt}
\setlength{\leftmargin}{\labelwidth}
\addtolength{\leftmargin}{\labelsep}
\setlength{\parindent}{0pt}
\setlength{\listparindent}{\parindent}
\setlength{\parsep}{1pt}
\setlength{\topsep}{1pt}}}{\end{list}}

\newtheorem{theorem}{Theorem}
\newtheorem{lemma}{Lemma}

\usepackage{amsmath}
\def\BibTeX{{\rm B\kern-.05em{\sc i\kern-.025em b}\kern-.08em
    T\kern-.1667em\lower.7ex\hbox{E}\kern-.125emX}}

\begin{document}

\title{\sysname: Breaking the Global Ordering Barrier in Multi-BFT Consensus}

\author{
	\IEEEauthorblockN{
		Hanzheng Lyu\IEEEauthorrefmark{1}, 
		Shaokang Xie\IEEEauthorrefmark{2}, 
		Jianyu Niu\IEEEauthorrefmark{3}, 
		Mohammad Sadoghi\IEEEauthorrefmark{2}, \\
        Yinqian Zhang\IEEEauthorrefmark{4}, 
        Cong Wang\IEEEauthorrefmark{3},
        Ivan Beschastnikh\IEEEauthorrefmark{5}, 
        Chen Feng\IEEEauthorrefmark{1} 
} 
\IEEEauthorblockA{University of British Columbia (\IEEEauthorrefmark{1}Okanagan Campus, \IEEEauthorrefmark{5}Vancouver Campus
), \IEEEauthorrefmark{2}University of California, Davis, \\ \IEEEauthorrefmark{3}City University of Hong Kong, \IEEEauthorrefmark{4}Southern University of Science and Technology}

\IEEEauthorblockA{\IEEEauthorrefmark{1}\{lyuhanzheng@gmail.com, chen.feng@ubc.ca\}, \IEEEauthorrefmark{2}\{skxie,msadoghi\}@ucdavis.edu} \IEEEauthorrefmark{3}\{njianyu@gmail.com, congwang@cityu.edu.hk\}, \IEEEauthorrefmark{4}yinqianz@acm.org,
\IEEEauthorrefmark{5}bestchai@cs.ubc.ca

}

\maketitle

\begin{abstract}
Multi-Byzantine Fault Tolerant (Multi-BFT) consensus, which runs multiple BFT instances in parallel, has recently emerged as a promising approach to overcome the leader bottleneck in classical BFT protocols. However, existing designs rely on a global ordering layer to serialize blocks across instances, an intuitive yet costly mechanism that constrains scalability, amplifies failure propagation,  and complicates deployment. In this paper, we challenge this conventional wisdom. We present \sysname, the first Multi-BFT consensus framework that eliminates global ordering altogether. \sysname introduces an object-centric execution model that partitions transactions by their accessed objects, enabling concurrent yet deterministic execution across instances. To ensure consistency, \sysname combines lightweight lock-based coordination with a deadlock resolution mechanism, achieving both scalability and correctness.
We implement \sysname and evaluate it on up to 128 replicas in both LAN and WAN environments. Experimental results show \sysname outperforms several state-of-the-art Multi-BFT protocols in the presence of a straggler. These results demonstrate strong consistency and high performance by removing global ordering, opening a new direction toward scalable Multi-BFT consensus design.
\end{abstract}


\section{Introduction} \label{sec:intro} 
Byzantine Fault Tolerant (BFT) consensus is a cornerstone of modern decentralized systems, powering diverse applications such as blockchains~\cite{9543565,QuestMarko,lens,gupta2019proof}, decentralized finance (DeFi)~\cite{wang2022bft}, and decentralized storage~\cite{zhu2020blockchain}.
Most classical BFT protocols, including PBFT~\cite{pbft1999} and Zyzzyva~\cite{kotla2007zyzzyva}, adopt a leader-based architecture, where a designated replica (also called the leader) proposes transactions 
and coordinates with others to reach agreement.
However, as system scale grows, this traditional leader-based design becomes a key performance bottleneck~\cite{gai2021dissecting, stathakopoulou2019mir, avarikioti2020fnf, stathakopoulou2022state, gupta2021rcc,amiri2024bedrock}: the leader’s coordination workload increases linearly with the number of replicas, making it the dominant factor limiting throughput and latency.

To address this scalability bottleneck, Multi-BFT consensus has emerged as a promising direction~\cite{stathakopoulou2019mir, avarikioti2020fnf, stathakopoulou2022state, gupta2021rcc, MIR-BFT}. Multi-BFT runs multiple leader-based consensus instances (\eg, PBFT) in parallel, as illustrated in \figref{fig:paradigm}. Client transactions are partitioned into disjoint buckets, each handled by a separate instance, enabling concurrent agreement among multiple leaders.
Each instance commits its own blocks locally, which are then merged through a global ordering phase to form a single, totally ordered ledger. This design significantly increases throughput by better utilizing available bandwidth and computational capacity.

\begin{figure}[t]
	\centering
    \includegraphics[width=0.85\linewidth]{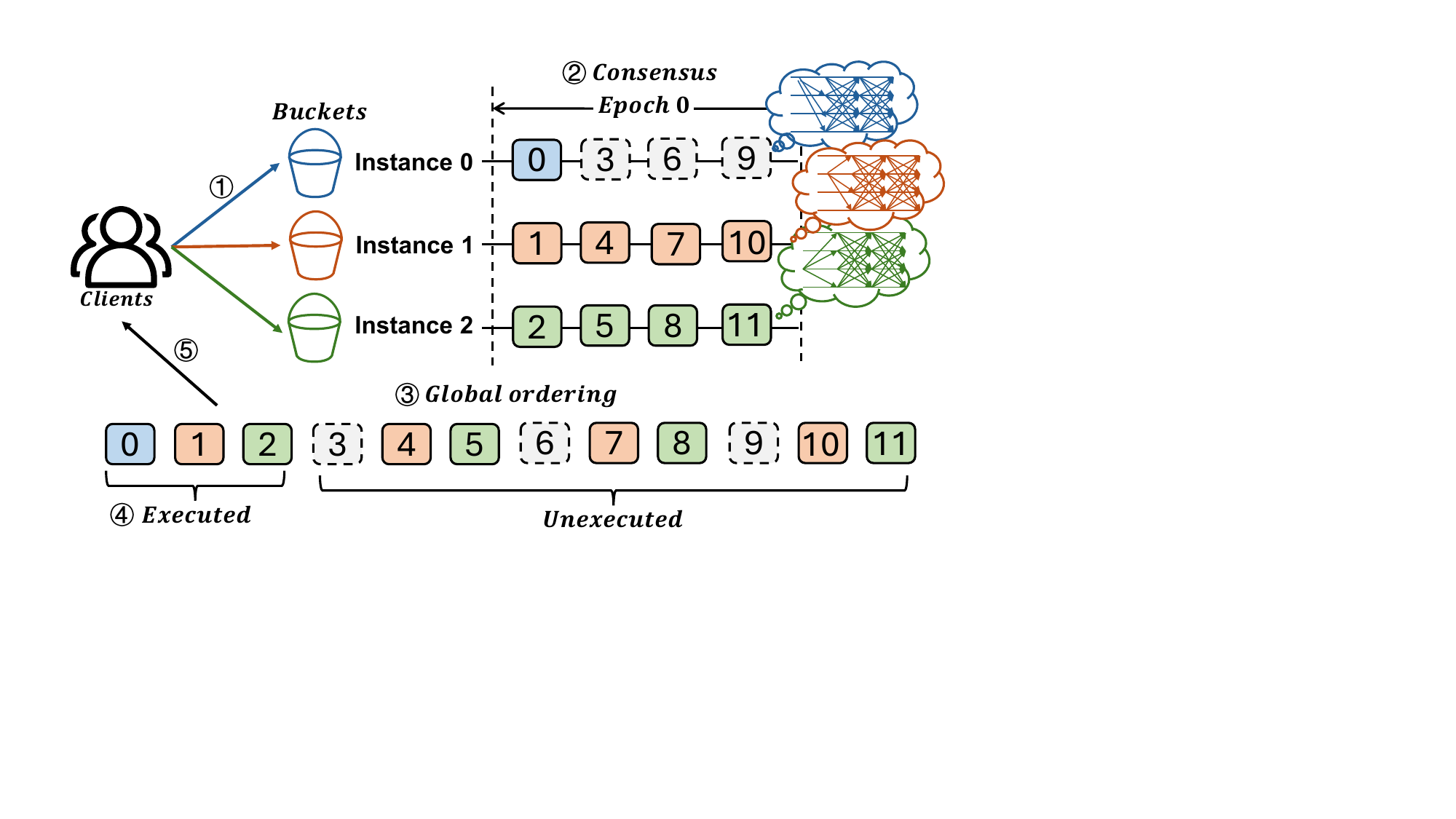}
	\caption{Multi-BFT consensus paradigm. Global ordering is the core of ordering blocks across instances.}
	\label{fig:paradigm}
    \vspace{-5mm}
\end{figure}

Despite its benefits, the global ordering layer--used to serialize blocks across instances--remains a major scalability barrier. While intuitively necessary for consistency, 
it introduces significant latency, particularly in the presence of slow instances. In practice, a single slow or faulty instance can stall the merging process, delaying confirmations for otherwise fast replicas. For example, in ISS~\cite{stathakopoulou2022state}, a state-of-the-art Multi-BFT protocol, the global ordering phase alone accounts for up to 92.8\% of total transaction latency with just one straggler among 16 replicas. Although recent optimizations such as Ladon~\cite{Ladon2025} and Orthrus~\cite{Orthrus2025} mitigate straggler effects, global ordering still dominates latency, contributing over 70\% of total delay in Orthrus. See more details of these limitations in \secref{sec:background}. 

Global ordering also introduces execution inefficiency. Once blocks are serialized, replicas typically execute transactions sequentially according to the global log. Yet, many workloads contain semantically independent transactions that operate on disjoint objects. Hence, these computation workloads can be safely executed in parallel. This creates a double serialization barrier: 1) the ordering phase serializes instance outputs; and then 2) sequential execution re-serializes independent transactions. As a result, current Multi-BFT systems fail to exploit inherent transaction-level parallelism.

These limitations motivate a central question:
\textit{Can Multi-BFT consensus preserve safety and deterministic execution without enforcing a costly global order?}

In this paper, we propose \sysname, the first Multi-BFT consensus framework that eliminates the need for global ordering. Instead of enforcing a total transaction order and executing afterward, \sysname unifies ordering and execution by leveraging transaction dependencies to expose concurrency while maintaining consistency.
The key insight is that transactions operate on objects, which are independently updatable state items (\eg, accounts or contract variables).
As long as operations on the same object are executed sequentially, the correctness of the overall system state is preserved. Unlike sharding---where {replicas} maintain only a subset of objects and must coordinate across partitions---each replica in \sysname retains a global view of all objects, enabling local coordination without expensive cross-instance protocols.

To realize this model, \sysname partitions the global object space into disjoint groups, each managed by a dedicated instance responsible for ordering and committing transactions on its assigned objects. Transactions are classified into two types: intra-instance transactions, which affect objects within a single group, and cross-instance transactions, which span multiple groups. Replicas execute instances concurrently: intra-instance transactions are executed locally, whereas cross-instance transactions are carried out by coordinating their sub-operations across relevant instances. This object partitioning introduces an new atomicity challenge: sub-operations distributed across instances must either all commit or all abort for consistency. 

\sysname addresses this challenge using a locking-based coordination mechanism: a transaction is prohibited from accessing objects involved in another in-flight transaction until the latter completes. Each transaction must acquire locks on all relevant objects before execution, guaranteeing that no intermediate state is exposed and that transaction outcomes remain atomic.
However, locking introduces the potential for deadlocks. 
Different instances may acquire locks in different orders for the same cross-instance transaction set. Such inconsistent lock acquisition orders can create cyclic waiting dependencies, preventing involved transactions from making progress.

To resolve this, \sysname implements a deadlock detection and resolution mechanism. When a transaction is blocked during lock acquisition, the system expands a distributed deadlock group by tracing dependency edges from all involved instances. If a dependency cycle is detected, \sysname deterministically aborts one or more transactions in the group—ensuring that all instances resolve the conflict consistently and the system continues to make progress.

We build an end-to-end prototype of \sysname in Go~\cite{golang} and conduct extensive experiments on AWS to evaluate its performance. 
We compare \sysname with the state-of-the-art Multi-BFT systems, including ISS~\cite{stathakopoulou2022state}, RCC~\cite{gupta2021rcc}, Mir-BFT~\cite{MIR-BFT}, DQBFT~\cite{dqbft}, Ladon~\cite{Ladon2025} and Orthrus~\cite{Orthrus2025}, in terms of throughput and latency.

\bheading{Contributions.} The main contributions are as follows.

\begin{packeditemize}
 
\item We propose \sysname, the first Multi-BFT consensus framework that eliminates the need for global ordering. \sysname significantly improves throughput and reduces end-to-end latency, especially in the presence of slow instances. 
 
\item We design a concurrent execution model for Byzantine environments, allowing replicas to execute transactions locally while preserving deterministic consistency. This model leverages lock-based execution and distributed deadlock detection to enable high-throughput parallel processing without violating correctness.
 
\item We build a prototype of \sysname and evaluate its performance in both WAN and LAN environments. Experimental results show substantial throughput improvement, with up to 9.0× higher throughput in WAN deployments and 7.4× higher throughput in LAN settings, compared to existing Multi-BFT protocols in the presence of a straggler.
\end{packeditemize}


\section{Motivations}\label{sec:background} 
\subsection{Anatomy of Multi-BFT Consensus} 
A Multi-BFT consensus system runs multiple BFT instances in parallel to improve throughput. Each transaction follows a typical five-stage process, as shown in \figref{fig:paradigm}: \ding{172} A client submits its transaction to replicas;  \ding{173} Replicas reach instance-level agreement to commit transactions using a BFT protocol; \ding{174} Replicas order committed transactions across instances into a global ledger; \ding{175} Replicas execute the ordered transactions; \ding{176} Replicas return results to the client.

The end-to-end latency and throughput are jointly determined by the following components: the \textit{transmission phase} (\ding{172} and \ding{176}), which is dominated by
network delay; the \textit{consensus phase} (\ding{173}), which depends on the underlying BFT protocol; and the \textit{global ordering phase} (\ding{174}), which is determined by the global ordering algorithm; and finally, the \textit{execution phase} (\ding{175}), whose performance depends on the system’s execution model, \ie, whether transactions are executed serially, speculatively, or in parallel based on their dependencies.

Among these components, the first two are well-studied and highly optimized in modern BFT systems, and such optimizations are largely orthogonal to the parallelization benefits brought by Multi-BFT designs. Thus, the major bottlenecks and opportunities for improvement lie in:
{(i) reducing or eliminating the global ordering cost}, and
{(ii) enabling highly parallel execution}. 
However, in practice, these two components fundamentally limit the scalability of Multi-BFT systems. We now examine these scalability challenges.

\subsection{Analyzing Scalability Challenges}
\bheading{Global ordering limits scalability.}
While parallel instances can improve consensus throughput, the system must still serialize their results through a global ordering phase to ensure consistent execution. This stage becomes the dominant bottleneck as it couples the progress of all instances: a single delayed or crashed instance can block the entire system from advancing. 
Specifically, when one instance lags behind others, the system cannot finalize subsequent transactions from faster instances because their assigned global indices remain unfilled. 

To understand the above limitation, \figref{fig:paradigm} illustrates a concrete case, in which instance 0 significantly lags behind the other instances: while the others have already produced four blocks, instance 0 has only produced one. This delay creates gaps at positions 3, 6, and 9 in the global log. As a result, the execution stalls after block 2, and subsequent blocks(\eg, blocks 7, 8, 10, and 11) cannot be executed until the missing blocks arrive.
This dependency chain stalls the global log and causes the end-to-end latency to grow dramatically, even though the majority of instances continue to make local progress.

Table~\ref{tab:breakdown} shows experimental results to evaluate the impact of global ordering phase, which exceeds 90\% of the total end-to-end latency in ISS, which is a state-of-the-art pre-determined global ordering Multi-BFT consensus.
Although recent works such as Ladon~\cite{Ladon2025} and Orthrus~\cite{Orthrus2025} adopt dynamic or hybrid ordering to mitigate blocking, they still retain a global ordering barrier. Consequently, Orthrus continues to spend approximately 70\% of its total latency in global ordering under the same conditions.
Prior studies~\cite{stathakopoulou2022state, Orthrus2025} also obtain similar findings that under certain conditions, the global ordering stage can dominate overall latency.

\begin{table}[t]
\centering
\caption{Latency breakdown (seconds) under one straggler among 16 replicas in a WAN environment.}
\label{tab:breakdown}
\setlength{\tabcolsep}{6pt} 
\begin{tabular}{lcccc}
\toprule
\textbf{Protocol} & \textbf{Transmission} & \textbf{Consensus} & \textbf{Ordering} & \textbf{Execution} \\
\midrule
ISS~\cite{stathakopoulou2022state}     & 0.177 & 2.502 & 34.513 & 0.551 \\
Orthrus~\cite{Orthrus2025}             & 0.176 & 2.553 & 7.430  & 0.558 \\
\bottomrule
\end{tabular}
\vspace{-4mm}
\end{table}

\bheading{Execution parallelism left untapped.}
Beyond ordering, existing Multi-BFT systems also underutilize execution parallelism. 
After transactions are globally ordered, replicas typically execute them sequentially to preserve determinism. 
This design overlooks the fact that many transactions are independent and can safely execute in parallel. 
Consequently, Multi-BFT consensus faces two layers of serialization: one at the ordering stage and another at the execution stage. 
This \textit{double serialization barrier} fundamentally limits scalability, even when the consensus layer scales linearly with the number of instances.

We now simply analyze the theoretical execution time under sequential and parallel models. 
Let $t$ denote the execution time of a single transaction. 
In the sequential model, execution time grows linearly with the number of transactions $N$, \ie, $T_{\text{seq}} = N \cdot t$. 
In contrast, parallel execution allows multiple transactions to be processed concurrently. 
With $k$ parallel execution units, the theoretical execution time becomes $T_{\text{par}} = \lceil N / k \rceil \cdot t$. 
Thus, increasing the degree of parallelism effectively mitigates the execution bottleneck.

\bheading{Summary.} This work is motivated by the observation that ordering and execution are inherently connected: the way transactions are ordered determines how they can be executed. By rethinking this relationship, we can design a system that parallelizes both dimensions simultaneously by removing the global ordering in Multi-BFT consensus.

\section{System Model and Goals} \label{sec:hydramodel}
\subsection{System Model}
We consider a system composed of $n = 3f+1$ replicas, collectively denoted as the set $\mathcal{N}$, responsible for processing transactions from a group of clients. We assume a subset of up to $f$ replicas as \textit{Byzantine}, represented as $\mathcal{F}$. Byzantine replicas may behave arbitrarily. The remaining replicas (in $\mathcal{N} \setminus \mathcal{F}$) are considered honest and strictly follow the protocol.
We assume a single, computationally bounded adversary that controls all Byzantine replicas and cannot break cryptographic primitives to falsify messages from honest replicas (with a negligible probability). 
Each replica maintains a public/private key pair for signing and verifying messages.

\bheading{Network model.} We assume each pair of honest replicas is connected by an authenticated and reliable communication link. We adopt the partial synchrony model proposed by Dwork \etal \cite{dwork1988consensus}, commonly used in BFT consensus \cite{pbft1999, hotstuff}. There is an established bound, denoted as $\Delta$, and an undefined Global Stabilization Time (\textsf{GST}). After the \textsf{GST} point, the delivery of any message transmitted between two honest replicas within the $\Delta$ limit is guaranteed. That is, the system behaves  \textit{synchronously} after the \textsf{GST}.

\subsection{Data Model} 
\bheading{Objects.} 
We adopt an object-centric design~\cite{birrell1984implementing,open1992introduction,wollrath1996distributed,redmond1997dcom,ostrowski2008programming,weihl1988commutativity,blackshear2023sui}, where objects serve as the fundamental data units processed within the system. These objects are persistent, similar to accounts, and are formally represented as a tuple:
    \( o = (\textit{key}, \textit{value})\),
where \(\textit{key}\) is a unique identifier, and \(\textit{value}\) represents the current state of the object, which can be updated through transactions. 
For example, in Ethereum, each account can be treated as a distinct object, where the account address is the key and its balance is the value.

To ensure consistency and prevent conflicts in concurrent execution, objects can be explicitly locked and unlocked during transaction processing. When a transaction intends to modify an object, it must first acquire a lock on the object, preventing other transactions from modifying it simultaneously. Once the transaction is confirmed, the lock is released, allowing subsequent transactions to access the object. 


\bheading{Transactions.} 
A transaction is structured as a Directed Acyclic Graph (DAG), formally defined as:
\( tx = (\textit{id}, V, E) \), where \(\textit{id}\) is a unique identifier for the transaction. The set \( V \) consists of vertices, where each vertex represents an operation on an object. Formally, a vertex \( v \) is expressed as:
\(v = (o, p, c, \sigma)\), where \( o \) is the object being modified, \( p \) denotes the operation applied to \( o \), \( c \) represents the conditions required to execute \( p \), and \( \sigma \) is the cryptographic signature of the owner of \( o \). The set \( E \) contains directed edges, where an edge \( e_{i,j} \in E \) signifies a dependency constraint, indicating that vertex \( v_i \) must be processed before vertex \( v_j \) can be executed.


\bheading{Block.} A block is defined as a tuple \( b = (txs, ins, sn, \sigma) \), where \( txs \) denotes a batch of transactions, and \( ins \) specifies the instance processing the block. The sequence number \( sn \) is the index of a block within its instance. Finally, \( \sigma \) represents the cryptographic signature on \( b \), guaranteeing both authenticity and integrity.

\subsection{Preliminaries}

\bheading{Sequenced broadcast (SB).}
SB is a variant of Byzantine total order broadcast that provides a consistent total ordering of messages among replicas.
In SB, a designated leader \textit{broadcasts} messages that are associated with monotonically increasing sequence numbers, while all replicas collaborate to \textit{deliver} these messages in a consistent global order.
It employs a failure detector to identify and handle faulty or silent leaders. SB guarantees two properties: 
\begin{packeditemize}
    \item \textit{Termination.} All honest replicas eventually deliver exactly one message for each sequence number. 

    \item \textit{Agreement.} The delivered messages are identical across all honest replicas

\end{packeditemize}
We adopt SB as a black-box abstraction for our consensus instances, where it accepts client transactions as input and produces a totally ordered stream of delivered transactions that are consistent among all honest replicas.


\subsection{System Goals}\label{subsec:hydra_goal}
We consider a Multi-BFT system composed of \( m \) BFT instances, indexed from \( 0 \) to \( m-1 \). Transactions are generated by clients and forwarded to replicas for processing, serving as the system’s input. 
Each instance operates in sequential rounds of a SB protocol, wherein a designated leader broadcasts a block of transactions and coordinates with backup replicas to deliver it. Transactions are executed after being partially ordered within an SB instance. A transaction is considered confirmed once it has been executed, regardless of whether the execution is successful or unsuccessful. The system state \( S \) is represented as a tuple, where each element corresponds to the maximum sequence number \( sn \) of an SB instance. Formally, the Multi-BFT system state is defined as:
\begin{equation}
    S = (sn_0, sn_1, \dots, sn_{m-1})
\end{equation}
where \( sn_i \) denotes the maximum sequence number for the SB instance indexed by \( i \).
The system must ensure two fundamental properties:

\begin{packeditemize}
    \item \textbf{Safety.} If two honest replicas reach the same state \( S \), they must have identical values for every object in \( S \). 

    \item \textbf{Liveness.} If a transaction \( tx \) is received by at least one honest replica, then the client will eventually receive a response for \( tx \).
\end{packeditemize}


\section{Design Rationale}\label{sec:rationale}
We present key insights of removing global ordering in Multi-BFT consensus, followed by challenges and corresponding solutions introduced by this new architectural design.

\subsection{Key Insights}
Traditional Multi-BFT systems enforce a global order on all transactions, although correctness only requires that all operations on the same object be processed in a consistent sequential order across replicas. Transactions accessing disjoint objects do not need global ordering and can safely execute in parallel.
This insight motivates an object-centric redesign of Multi-BFT architecture: by maintaining per-object ordering, \sysname preserves deterministic execution while eliminating global ordering, thereby unlocking significantly higher concurrency.


\bheading{A concrete example.}
To illustrate the idea, consider a transaction where user $A$ transfers 10 tokens to user $B$.
This transaction consists of two operations: $A.value{-}10$ and $B.value{+}10$, 
each targeting a distinct object.
In conventional Multi-BFT systems, both operations are encapsulated in a single transaction that must be totally ordered with all others to ensure deterministic replay.
In contrast, \sysname\ allows these operations to be processed without global ordering.
If both objects $A$ and $B$ reside in the same instance, the transaction follows that instance’s local order; if they are placed in different instances, each operation is ordered locally within its own instance.

\subsection{Challenges and Solutions}
Decomposing a transaction across multiple instances introduces two key challenges.

\bheading{Challenge 1: Atomicity across instances.}
Suppose instance~0 executes $B.value{+}10$ while instance~1 later rejects $A.value{-}10$ due to insufficient balance.
The resulting state is inconsistent---$B$ gains tokens that were never deducted from $A$.
A straightforward remedy would be to roll back $B.value{+}10$ when $A.value{-}10$ fails, but this is only safe if $B$’s updated balance has not been used in subsequent transactions.
Otherwise, reverting $B.value{+}10$ could invalidate dependent operations, leading to cascading inconsistencies.

To prevent such inconsistencies, \sysname adopts a lock-based execution model inspired by classical concurrency control mechanisms such as Two-Phase Locking (2PL)~\cite{eswaran1976notions,bernstein1987concurrency,gray1992transaction}. Each transaction must acquire locks on all of its objects before execution.
Only after obtaining all locks is the transaction executed. After that, all acquired locks in the transaction are released. This ensures that no transaction is ever partially executed: either all operations complete atomically or none take effect. 

\bheading{Challenge 2: Conflicting ordering among instances.}
Different instances may impose different orders on the same set of transactions, creating deadlocks.
For example, consider two transactions: $tx_1: A \rightarrow B$ and $tx_2: B \rightarrow A$.
If instance~0 locks $A$ for $tx_1$ while instance~1 locks $B$ for $tx_2$, both transactions will wait indefinitely for the other’s lock.
To resolve deadlocks, \sysname\ integrates a deadlock detection and resolution mechanism~\cite{bernstein1987concurrency}:
once a transaction is confirmed but blocked on lock acquisition, the system searches for cycles in the wait-for graph and deterministically aborts one or more transactions to break the cycle.

These two mechanisms jointly ensure that all operations on the same object are guaranteed to be executed in a consistent and sequential order across replicas. It ensures that no transaction is ever partially executed---each transaction either completes all its intended operations atomically or is entirely aborted without leaving side effects. Importantly, this design allows the system to make progress with high concurrency.

\begin{figure}[t]
	\centering
    \includegraphics[width=0.9\linewidth]{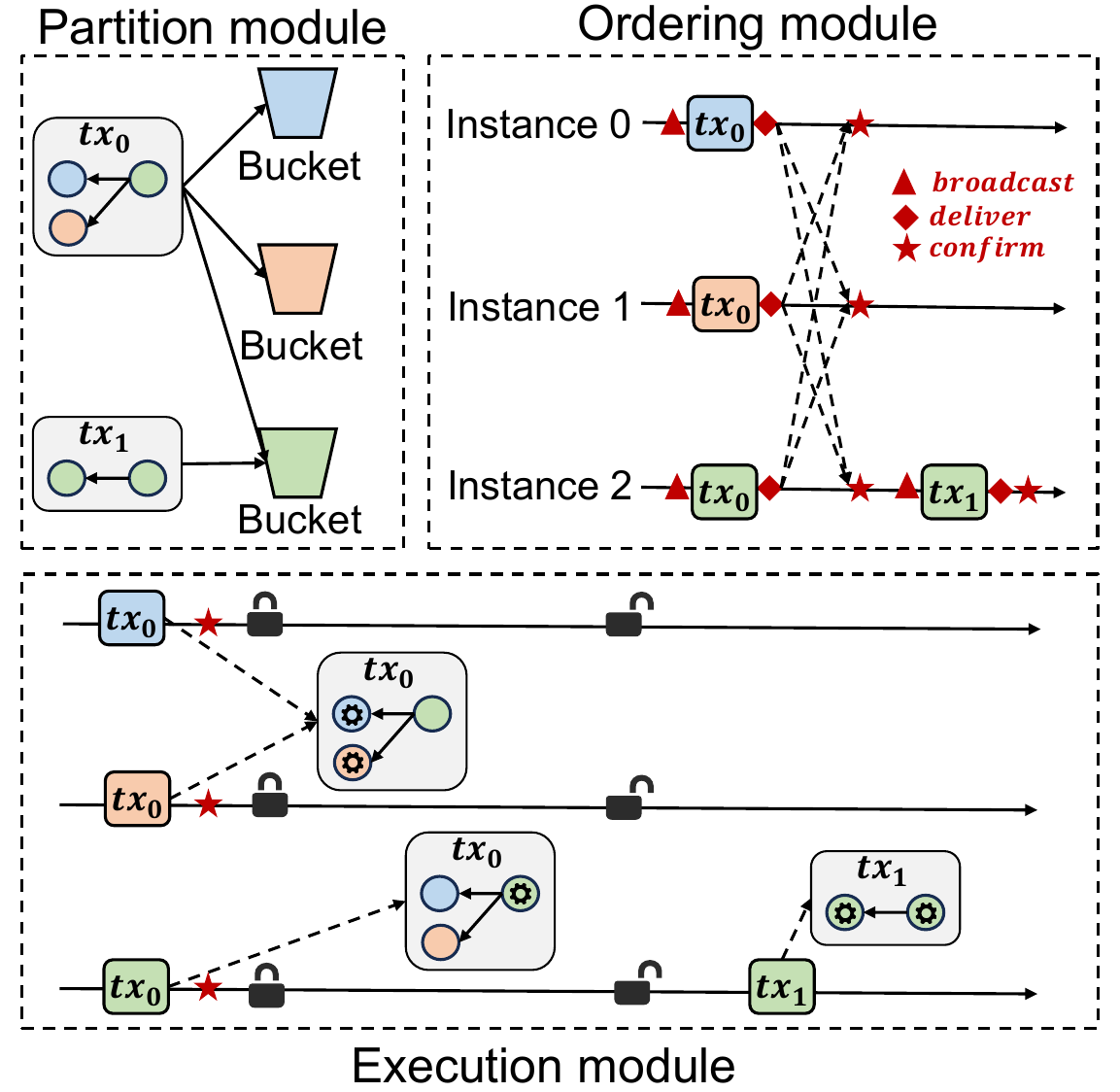}
	\caption{Overview of \sysname. \sysname contains three key steps: transaction partition, transaction ordering, and transaction execution.}
	\label{fig:overview}
    \vspace{-4mm}
\end{figure}

\section{System Design} \label{sec:hydradesign}
We present \sysname, which contains three key steps: transaction partition (\secref{sec:hydrapartitionAlgorithm}), transaction ordering (\secref{sec:hydraorderAlgorithm}), and transaction execution (\secref{sec:hydraexecuteAlgorithm}). Except for these transaction processes, \sysname also performs checkpointing and garbage collection (\secref{sec:hydracheckpoint}) to ensure efficient resource management and state consistency.

\subsection{Architectural Overview}
\figref{fig:overview} presents an architectural overview of \sysname, which operates multiple instances of the SB protocol in parallel. 
Transactions are assigned to buckets based on the objects they contain, and each bucket is mapped to a specific instance.  
\sysname confirms transactions through the following phases:  
\begin{packeditemize}
    \item \textbf{Transaction partition}: Each transaction is assigned to the corresponding instances based on the objects it involves. If all objects within a transaction belong to the same instance, it is classified as an \textit{intra-instance transaction}; otherwise, it is a \textit{cross-instance transaction}. 
    
    \item \textbf{Transaction ordering}: Replicas execute SB protocols within their respective instances to order transactions. A transaction is considered confirmed once all its corresponding instances have successfully delivered it.  
    
    \item \textbf{Transaction execution}: Replicas execute transactions by respecting both the order (sequence numbers) and the dependencies (represented by the DAG). 
\end{packeditemize}

\bheading{A concrete example.} \figref{fig:overview} illustrates both cross-instance and intra-instance transactions. 
Consider a transaction $tx_0$ that involves three objects distributed across three different instances. The objects in the first layer are assigned to instance 0 and instance 1, while the object in the second layer belongs to instance 2. The execution proceeds as follows: all three objects are proposed and reach consensus concurrently within their respective instances. Once all involved objects have been delivered, $tx_0$ is considered confirmed. 

The replica then acquires locks on all related objects, ensuring atomicity. Execution begins with the first-layer objects in instance 0 and instance 1, followed by the second-layer object in instance 2, respecting the topological order of the transaction’s internal DAG.
In contrast, transaction $tx_1$ is an intra-instance transaction that involves two objects, both assigned to instance 2. Since all related objects are processed within the same instance, $tx_1$ is confirmed immediately upon delivery in instance 2. Once confirmed, the transaction is executed directly according to its internal DAG, without requiring coordination across instances.

\subsection{Transaction Partition}\label{sec:hydrapartitionAlgorithm}
Algorithm~\ref{algorithm:hydramain} shows the detailed process to partition transactions. 
Upon receiving a transaction \( tx \) (Line 11), replica \( r \) assigns \( tx \) to the appropriate buckets based on the objects it involves. The process begins with verifying the transaction’s format and checking the owner's signature for authenticity (Line 12). 
Once validated, replica \( r \) invokes the \( \mathsf{assign}\) function to determine the corresponding bucket(s) for \( tx \) (Line 14). Specifically, the replica obtains the appropriate bucket index for each object in the transaction. The transaction is then appended to the corresponding bucket (Line 15). Each bucket operates as an append-only list for backup replicas but supports both push and pull operations for the leader. To prevent duplication, if \( tx \) is already present in a bucket, it will not be added again.

\subsection{Transaction Ordering}\label{sec:hydraorderAlgorithm}
Each replica maintains a \( log \) for each SB instance, which facilitates the partial ordering of transactions within that instance. The log comprises multiple entries, each capable of storing a batch of transactions.  
The SB protocol defines two key events: \( \langle\mathsf{sb\text{-}broadcast}|b\rangle \) and \( \langle\mathsf{sb\text{-}deliver}|b\rangle \). The event \( \langle\mathsf{sb\text{-}broadcast}|b\rangle \) occurs when a block \( b \) is broadcasted within the SB, while \( \langle\mathsf{sb\text{-}deliver}|b\rangle \) denotes the successful ordering and delivery of block \( b \).

\bheading{Broadcast transactions.} 
As shown in Algorithm~\ref{algorithm:hydramain}, upon initializing \sysname (Line 1), if replica \( r \) is the leader of \( instance_i \), it enters a loop where it creates and broadcasts a block \( b \) for each sequence number \( sn \) in the current epoch (Lines 3-4).  
In each iteration, the leader pulls a specified number of the oldest transactions from the instance’s bucket (Line 5). If there are insufficient transactions to meet this threshold, the leader waits for a timeout. If the threshold is still not met after the timeout, it pulls all available transactions. It then assigns the instance index and sequence number to the block, signs it, and broadcasts it within the SB instance (Lines 6-7). All replicas participate in the ordering and delivery of the block.  
If \( r \) is not the leader, it acts as a backup, assisting in the ordering and delivery process within the SB instance.

\begin{algorithm}[t]
\caption{\sysname main algorithm on replica $r$}
\label{algorithm:hydramain}
\begin{algorithmic}[1]

\Upon{initialize system}
  \For{$i \in [0, m-1]$} \Comment{\textcolor{purple}{$m$ is the number of instances}}
    \If{$\mathsf{isLeader}(i, r)$} \Comment{\textcolor{purple}{$r$ is $instance_i$'s leader}} 
      \For{$sn \in \{0,1,2,\dots$\}}
        \State $b.\mathit{txs} \gets \mathsf{pullValidTx}(bucket_i)$
        \State $b.\mathit{ins} \gets i$, $b.\mathit{sn} \gets sn$, $b.\sigma \gets \mathsf{sign}(b, r)$
        \State \textbf{trigger} $\langle\mathsf{sb\text{-}broadcast}|b\rangle $
      \EndFor
    \EndIf
  \EndFor
\EndUpon

\Statex \Comment{\textcolor{purple}{transaction partition~~~~~~~~~~~~~~~~~~~~~~~~~~~~~~~~~~~~~~~~~~}} 
\Upon{receive $tx$}
  \If{$\mathsf{validateTx}(tx)$} \Comment{\textcolor{purple}{check format and signature}}
    \ForAll{$v \in tx.V$}
      \State $i \gets \mathsf{assign}(v.o)$  \Comment{\textcolor{purple}{map object to an instance}} 
      \State $\mathsf{push}(tx, bucket_i)$
    \EndFor
  \EndIf
\EndUpon

\Statex\Comment{\textcolor{purple}{transaction ordering~~~~~~~~~~~~~~~~~~~~~~~~~~~~~~~~~~~~~~~~~~}} 
\Upon{event $\langle \text{sb-deliver} \mid b \rangle$}
  \State $log[b.\mathit{ins}][b.\mathit{sn}] \gets b$ 
  \ForAll{$tx \in b$}
    \ForAll{$v \in tx.V \ \textbf{where}\ \mathsf{assign}(v.o) = b.\mathit{ins}$}
      \State $\mathsf{delivered}(v) \gets \mathsf{true}$
    \EndFor
    \If{$\forall w \in tx.V:\ \mathsf{delivered}(w)$}
      \State $\mathsf{confirm}(tx)$
    \EndIf
  \EndFor
\EndUpon

\end{algorithmic}
\end{algorithm}

\bheading{Deliver transactions.} 
As shown in Algorithm~\ref{algorithm:hydramain}, upon delivering a block \( b \) from an SB instance  (Line 18), the replica orders the block by appending it to the \( log \) of the instance indexed \( b.ins \) at sequence number \( b.sn \)  (Line 19). Then, for each transaction \( tx \) contained in \( b \), the replica iterates over all vertices \( v \) in \( tx \).  If the object $v.o$ is assigned to the current instance $b.ins$, the replica marks \( v \) as delivered  (Lines 21-23). Once all vertices in \( tx \) are delivered, the transaction \( tx \) is confirmed (Lines 24-26).

\bheading{Failure detector.} 
In \sysname, a failure detection module is integrated into the SB protocol to ensure liveness in the presence of faulty leaders. This mechanism enables replicas to replace a misbehaving or crashed leader for each SB instance, and has been widely adopted in prior BFT systems~\cite{castro2002practical, hotstuff, FireLedger, gueta2019sbft}.  
For example, in PBFT, when replicas suspect the leader of Byzantine behavior, they initiate a view change to elect a new one. Similarly, in \sysname, when the leader \( l_i \) of instance \( \mathit{instance}_i \) fails at sequence number \( sn \), all honest replicas detect the failure, reach agreement on the state of \( \mathit{instance}_i \) and the new leader \( l'_i \), and then restart the instance from sequence number \( sn \) under the new leader.

\begin{algorithm}[t]
\caption{Execute transactions on replica $r$}
\label{algorithm:hydraexe}
\begin{algorithmic}[1]

\Upon{$\mathsf{firstPending}(log[i]) \neq \bot$}
\Statex \Comment{\textcolor{purple}{$log[i]$ denotes the ordered transaction log of $instance_i$}}
  \State $tx \gets \mathsf{firstPending}(log[i])$ 
  \ForAll{$v \in tx.V$}
    \If{$\mathsf{assign}(v.o) = i$}
      \State $\mathsf{lock}(v.o)$
    \EndIf
  \EndFor
  \If{$\forall v \in tx.V:\ \mathsf{locked}(v.o)$}
    \ForAll{$v \in \mathsf{toposort}(tx.V)$}
      \State $\mathsf{execute}(tx)$ 
       \Comment{\textcolor{purple}{execute vertices by dependency}}
    \EndFor
    \If{$\exists v \in tx.V : \mathsf{failed}(v)$} 
      \State $\mathsf{rollback}(tx)$ \Comment{\textcolor{purple}{undo all executed vertices in $tx$}}
      \State $\mathsf{replyToClient}(tx,\ \textsc{failure})$
    \Else
      \State $\mathsf{replyToClient}(tx,\ \textsc{success})$
    \EndIf
    \ForAll{$v \in tx.V$}
      \State $\mathsf{unLock}(v)$
    \EndFor
  \EndIf
\EndUpon
\end{algorithmic}
\end{algorithm}

\begin{algorithm}[ht]
\caption{Deadlock resolution on replica $r$}
\label{algorithm:deadlock}
\begin{algorithmic}[1]

\Upon{$\mathsf{confirmed}(tx)\ \wedge\ \mathsf{isInterTx}(tx)$}
  \State $D \gets \mathsf{ExpandDeadlockGroup}(tx)$
  \State $victimList \gets \mathsf{SelectVictims}(D)$
  \ForAll{$tx \in victimList$}
    \State $\mathsf{abort}(tx)$
      \State $\mathsf{reAdd2Buckets}(tx)$  
    \EndFor
\EndUpon

\Statex \Comment{\textcolor{purple}{search transactions forming a deadlock with $tx$~~~~~~~~~}}
\Function {$\mathsf{ExpandDeadlockGroup}(tx)$}{}
  \State $D \gets \{tx\}$
  \Repeat
    \State $D_{\text{prev}} \gets D$
    \ForAll{$t \in D_{\text{prev}}$}
      \ForAll{$t' \in \mathsf{findOrderingConflicts}(t)$}
        \State $D \gets D \cup \{t'\}$
      \EndFor
    \EndFor
  \Until{$D = D_{\text{prev}}$}
  \State \Return $D$
\EndFunction

\Statex \Comment{\textcolor{purple}{select victim transactions in $D$ to eliminate deadlocks~~}}
\Function {$\mathsf{SelectVictims}(D)$}{}
  \State $victims \gets \emptyset$
  \While{$\mathsf{hasDeadlock}(D)$}
    \State $victim \gets \text{transaction in } D \text{ with smallest index}$
    \State $victims \gets victims \cup \{victim\}$
    \State $D \gets D \setminus \{victim\}$
  \EndWhile
  \State \Return $victims$
\EndFunction

\Statex \Comment{\textcolor{purple}{find transactions ordered inconsistently with $tx$ ~~~~~~~~}}
\Function{$\mathsf{findOrderingConflicts}(tx)$}{} 
  \State $C \gets \emptyset$
  \ForAll{$(i, j) \in \mathsf{instances}(tx),\ i \neq j$}
    \ForAll{$tx' \in \mathsf{prefix}[i](tx)$}
      \If{$tx' \in \mathsf{suffix}[j](tx)$}
      \Statex \Comment{\textcolor{purple}{$tx'$ is ordered after $tx$ in $j$}}
        \State $C \gets C \cup \{tx'\}$ 
      \ElsIf{$tx' \notin \mathsf{prefix}[j](tx) \land j \in \mathsf{instances}(tx')$} \Comment{\textcolor{purple}{$tx'$ pending in $j$, will be ordered after $tx$ in $j$}}
        \State $C \gets C \cup \{tx'\}$ 
      \EndIf
    \EndFor
  \EndFor
  \State \Return $C$
\EndFunction

\end{algorithmic}
\end{algorithm}

\subsection{Transaction Execution}\label{sec:hydraexecuteAlgorithm} 

\sysname’s execution model is inspired by well-established ideas in distributed databases. The requirement that a transaction acquires all necessary object locks before execution closely echoes 2PL~\cite{eswaran1976notions,gray1992transaction,bernstein1987concurrency}, which ensures atomicity and isolation by preventing other transactions from accessing locked resources. Similarly, the deadlock resolution mechanism parallels classical deadlock detection techniques that rely on identifying cycles in wait-for graphs~\cite{bernstein1987concurrency,kaveh2001deadlock}. 

As shown in Algorithm~\ref{algorithm:hydraexe}, when a replica identifies the first pending transaction $tx$ in $log[i]$ (Lines 1-2), it iterates over all vertices $v \in tx.V$. For each vertex $v$, if its associated object $v.o$ is assigned to the current instance, the replica attempts to acquire a lock on $v.o$ (Lines 3-6). Once all required objects are successfully locked (Line 8), the replica executes the transaction according to the topological order of the DAG (Lines 9-11). If any vertex fails during execution, the replica rolls back all successfully executed vertices and replies to the client with a failure (Lines 12-14). Otherwise, if all vertices are executed successfully, the replica replies with a success (Line 16). Finally, the replica releases all locks held by $tx$ (Lines 18-20), enabling subsequent transactions to access those objects.

Note that the algorithm does not explicitly handle the case where $tx$ appears in only a subset of its target instances. In such cases, if $tx$ remains unconfirmed at the end of an epoch, the replica will abort $tx$ and release any locks it may have acquired, and a failure response is sent back to the client.

\bheading{Deadlock resolution.} Deadlocks may occur in this execution model, as transactions can appear in different orders across multiple instances, leading to conflicting execution dependencies. This is a well-known issue in concurrent systems, especially when each instance executes transactions independently based on local orderings. To address this, we adopt a deterministic deadlock resolution mechanism, as shown in Algorithm~\ref{algorithm:deadlock}.

When a cross-instance transaction $tx$ is confirmed (Line 1), \ie, it has been delivered by all the instances it is assigned to, the replica initiates deadlock detection by invoking $\mathsf{expandDeadlockGroup}$ (Line 2). This procedure recursively constructs a set of transactions that may be involved in a cycle with $tx$ due to inconsistent instance-level orderings (Lines 13-15). Specifically, the replica locates the position of $tx$ in the log of the corresponding instance (Line 31), and scans backwards to identify earlier pending transactions (Line 32). If such a transaction $tx'$ is found to appear after $tx$ in another instance to which both are assigned (Line 33), or if $tx'$ has not yet been ordered relative to $tx$ in another instance (Line 35), the two are considered to have an inconsistent ordering. 
The replica then adds $tx'$ to the conflict set $C$ (Line 36). This ensures that both explicit and potential future inconsistencies are captured, even when $tx'$ has not yet appeared in all relevant instance logs.  
Among these, it selects transactions that also access other objects in $tx$ and appear after $tx$ in those corresponding logs---indicating a conflicting relative order (Line 13). These transactions are added to the deadlock group (Line 14), and the process repeats until no new transactions are included (Line 17).

Once the deadlock group $D$ is constructed, the replica calls $\mathsf{selectVictims}$ (Line 3) to determine which transactions should be aborted to break the cycle. This is done by iteratively removing the transaction with the smallest index (the index can be the hash of the transaction) from the group until the remaining transactions have no deadlock (Lines 22-26). The selected victims are then aborted and re-added to buckets again (Lines 4-7). This process ensures that deadlocks are resolved deterministically and consistently across replicas, without requiring additional consensus rounds.

\subsection{Checkpoint and Garbage Collection}\label{sec:hydracheckpoint}
\sysname employs a checkpoint protocol at the end of each epoch to facilitate state pruning and garbage collection. Once an epoch completes, each replica broadcasts a checkpoint message containing a signed digest summarizing the blocks it has processed during that epoch. Upon collecting at least $2f+1$ matching checkpoint messages, a replica establishes a stable checkpoint. This checkpoint allows replicas to safely discard transactions, whether successfully executed or aborted, from the completed epoch to reduce storage overhead.

\section{CORRECTNESS AND ROBUSTNESS ANALYSIS} \label{sec:hydracorrect} 
\subsection{Safety and Liveness Analysis}
We provide the proof sketch for safety and liveness (defined in \secref{subsec:hydra_goal}) here, while deferring the full proofs to Appendix A \cite{hydralong} due to space constraints.

\bheading{Safety.} Safety in \sysname follows from two key invariants. 
First, at the same delivery frontier, all honest replicas observe identical delivered transactions from each instance.
This is guaranteed by the Agreement property of the underlying SB protocol.
Second, deadlock detection and abort decisions are deterministic functions of these delivered transactions. Given the same set of delivered instance logs, all honest replicas compute the same deadlock group by recursively expanding conflicts.
Abort decisions are then derived from this deadlock group using a deterministic victim selection rule. As a result, all honest replicas abort the same set of transactions and commit the same remaining ones.
 
\bheading{Liveness.} Liveness requires that every transaction submitted by a correct client eventually receives a response.
In \sysname, this is ensured by the Termination property of the SB protocol and the epoch-based execution model. Specifically, any transaction proposed by an honest leader is eventually delivered by the SB protocol. Transactions that cannot be committed due to conflicts or deadlocks are deterministically aborted and reinserted for execution in subsequent epochs. The epoch mechanism bounds lock holding time and prevents transactions from being blocked indefinitely.


\subsection{Impact of Byzantine Behaviors}
\label{sec:byzantine}
We analyze how \sysname behaves under Byzantine behaviors that deliberately target its locking mechanisms and cross-instance coordination. 

\bheading{Lock manipulation by Byzantine leaders.}
A Byzantine leader could manipulate locking behavior from two aspects.

\iheading{(i) \textit{Delaying locked trasactions}.} A Byzantine leader may intentionally delay proposing a cross-instance transaction in one instance after it has already been proposed and has acquired locks in another instance.
For example, for a transaction spanning instances $ins_1$ and $ins_2$, a Byzantine leader in $ins_2$ may withhold or delay its proposal, causing locks acquired in $ins_1$ to be held for an extended period.
This behavior increases blocking for other transactions accessing the same objects and degrades throughput. 

\iheading{(ii) \textit{Creating deadlocks.}} A Byzantine leader may also attempt to induce deadlocks by proposing conflicting orderings across instances.
For instance, two cross-instance transactions $tx_1$ and $tx_2$ may be ordered as $tx_1 \rightarrow tx_2$ in $ins_1$, while a Byzantine leader in $ins_2$ intentionally proposes them in the reverse order $tx_2 \rightarrow tx_1$, creating a deadlock.

\bheading{Adversarial transaction injection by malicious clients.}
Malicious clients may deliberately generate multiple transactions involving the same set of objects and submit them to different instances in adversarial orders, intentionally creating deadlocks across instances and amplifying the likelihood of cross-instance deadlocks.

\bheading{Impact analysis of Byzantine behaviors.}
The above adversarial behaviors target the liveness of the system rather than safety. Byzantine leaders or malicious clients may delay locked transactions or induce cross-instance deadlocks, thereby degrading throughput and increasing latency. 
However, the impact of such adversarial behaviors is bounded.
First, deadlock resolution ensures that any induced deadlock is consistently identified system-wide.  Once detected, the deadlock resolution protocol deterministically aborts a bounded set of transactions, preventing indefinite blocking. Second, the epoch-based execution model enforces an upper bound on lock holding time: transactions that fail to make progress are aborted and retried in subsequent epochs, eliminating the possibility of persistent lock monopolization by Byzantine behavior.
As a result, while adversaries can temporarily reduce system efficiency, they cannot cause unbounded delays, inconsistent execution, or divergence among honest replicas.

\section{Performance Evaluation} \label{sec:hydraevaluation} 
In this section, we evaluate the performance of \sysname and compare it against other Multi-BFT protocols: Ladon~\cite{Ladon2025}, Orthrus~\cite{Orthrus2025}, ISS~\cite{stathakopoulou2022state}, RCC~\cite{gupta2021rcc},  Mir~\cite{MIR-BFT}, and DQBFT~\cite{dqbft}. 
We implemented \sysname in Go~\cite{golang}, using the PBFT consensus protocol~\cite{castro2002practical} as SB instances. 
Our evaluation aims to answer the following research questions:

\begin{packeditemize}
     \item \textbf{Q1:} How does \sysname perform compared to other Multi-BFT protocols? (\secref{sec:hydracompare})
     \item \textbf{Q2:} How does \sysname perform in varying proportions of cross-instance transactions? (\secref{sec:hydraproportion})
    \item \textbf{Q3:} How does \sysname perform under faults? (\secref{sec:hydrafaults}) 
    \item \textbf{Q4:} {How does \sysname’s memory usage scale with the number of instances? (\secref{sec:mem}) }
\end{packeditemize}

\subsection{Experimental Setup}\label{sec:expset}
We deploy our protocols on AWS EC2 instances (c5a.2xlarge) with one instance per node. Each instance  is equipped with 8vCPUs, 16GB RAM, and runs Ubuntu Linux 22.04. Experiments are conducted in both LAN and WAN environments. In the LAN setting, machines communicate over a private network with 1 Gbps bandwidth. For the WAN setup, machines are distributed across four Amazon cloud regions (France, North America, Australia, and Japan) with both public and private network interfaces limited to 1Gbps. We use NTP for clock synchronization.

Each replica acts as the leader for one instance and as a backup for the others, \ie, $m=n$.  
To maximize throughput, we use a large batch size of 4096 transactions, with each transaction carrying a 500-byte payload. We evaluate system performance under two network conditions: one with uniform performance across all instances, and another with a straggler scenario where one instance operates at one-tenth the speed of the others. Each experiment is repeated five times, and we report the median of the results.

We evaluate two key performance metrics: (1) throughput, defined as the number of transactions successfully responded to clients per second, and (2) latency, measured as the average end-to-end delay from the moment clients submit transactions until they receive \( f{+}1 \) responses from replicas. We report the peak throughput in kilo-transactions per second (ktps) before reaching saturation, along with the corresponding latency in seconds (s).

\subsection{Performance Comparison}\label{sec:hydracompare}
\figref{fig:performancewan} and \figref{fig:performancelan} compare the throughput and latency of \sysname and other Multi-BFT protocols with one straggler and without stragglers. We evaluate the throughput and latency with 8--128 replicas in both LAN and WAN environments. We only report results with a single straggler, since performance is primarily limited by the slowest replica and adding more stragglers does not significantly change throughput or latency~\cite{Ladon2025}.
For clarity, the compared protocols are categorized into three groups: 
\textit{pre-determined ordering schemes}, including ISS~\cite{stathakopoulou2022state}, Mir-BFT~\cite{MIR-BFT} and RCC~\cite{gupta2021rcc}, which enforce a fixed global transaction order;  
\textit{dynamic ordering schemes}, including DQBFT~\cite{dqbft} and Ladon~\cite{Ladon2025}, which adapt ordering based on instance progress to mitigate straggler effects; and \textit{hybrid ordering scheme} represented by Orthrus~\cite{Orthrus2025}, which partially relaxes global ordering by leveraging fast-path confirmation for certain transactions.
Following prior studies~\cite{zhang2023txallo}, we set the ratio of cross-instance transactions to 12\%.

\bheading{Performance in WAN.}
\figref{fig:performancewan} compares \sysname and other Multi-BFT protocols in WAN, with and without stragglers.  
With a straggler present, \sysname significantly outperforms pre-determined ordering schemes: with 128 replicas, throughput improves by approximately 9.0×.
This is because in pre-determined ordering schemes, a slow instance delays the entire system, amplifying straggler impact.
Dynamic (DQBFT, Ladon) and hybrid (Orthrus) ordering schemes alleviate this coupling but still require a global ordering phase. {For DQBFT, throughput decreases with more replicas because its single ordering instance becomes a bottleneck as the system scales.} Consequently, \sysname continues to achieve superior performance.
With 8 replicas, it achieves 47\% and 38\% higher throughput than Ladon and Orthrus, respectively, and with 128 replicas, \sysname surpasses DQBFT by 63\%.

Without stragglers, \sysname achieves comparable throughput to ISS, RCC, Ladon, and Orthrus, while consistently maintaining the lowest or near-lowest latency. 
This demonstrates that eliminating global ordering and incorporating deadlock resolution introduces negligible additional overhead, and scalability remains robust across different replica configurations.

\begin{figure}[t]
\vspace{-4mm}
    \centering

    \begin{subfigure}[t]{0.24\textwidth}
        \centering
        \includegraphics[width=\textwidth]{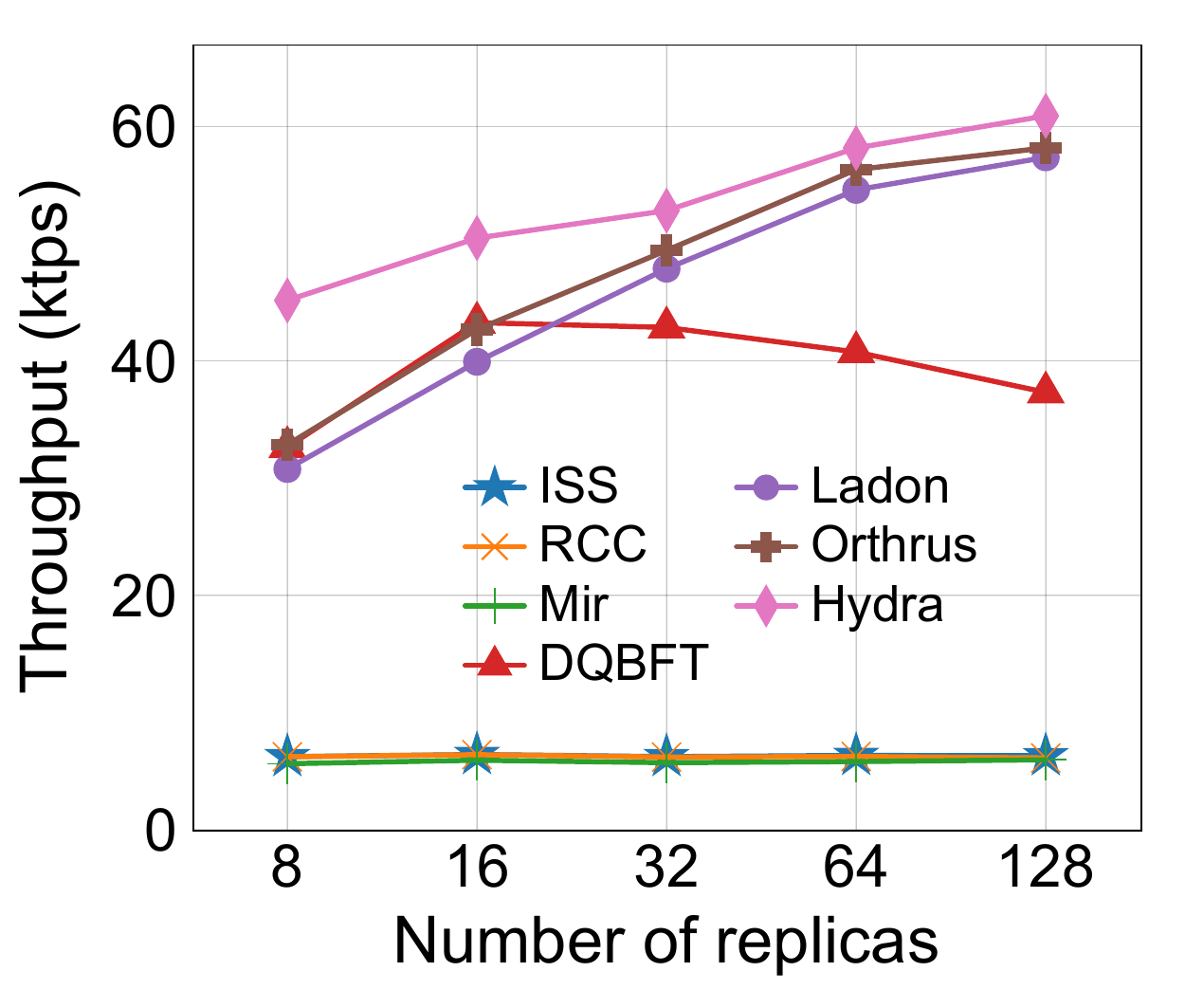}
        \caption{$\#$Straggler = 1, WAN}
        \label{fig:wan1}
    \end{subfigure}
    \hfill
    \begin{subfigure}[t]{0.24\textwidth}
        \centering
        \includegraphics[width=\textwidth]{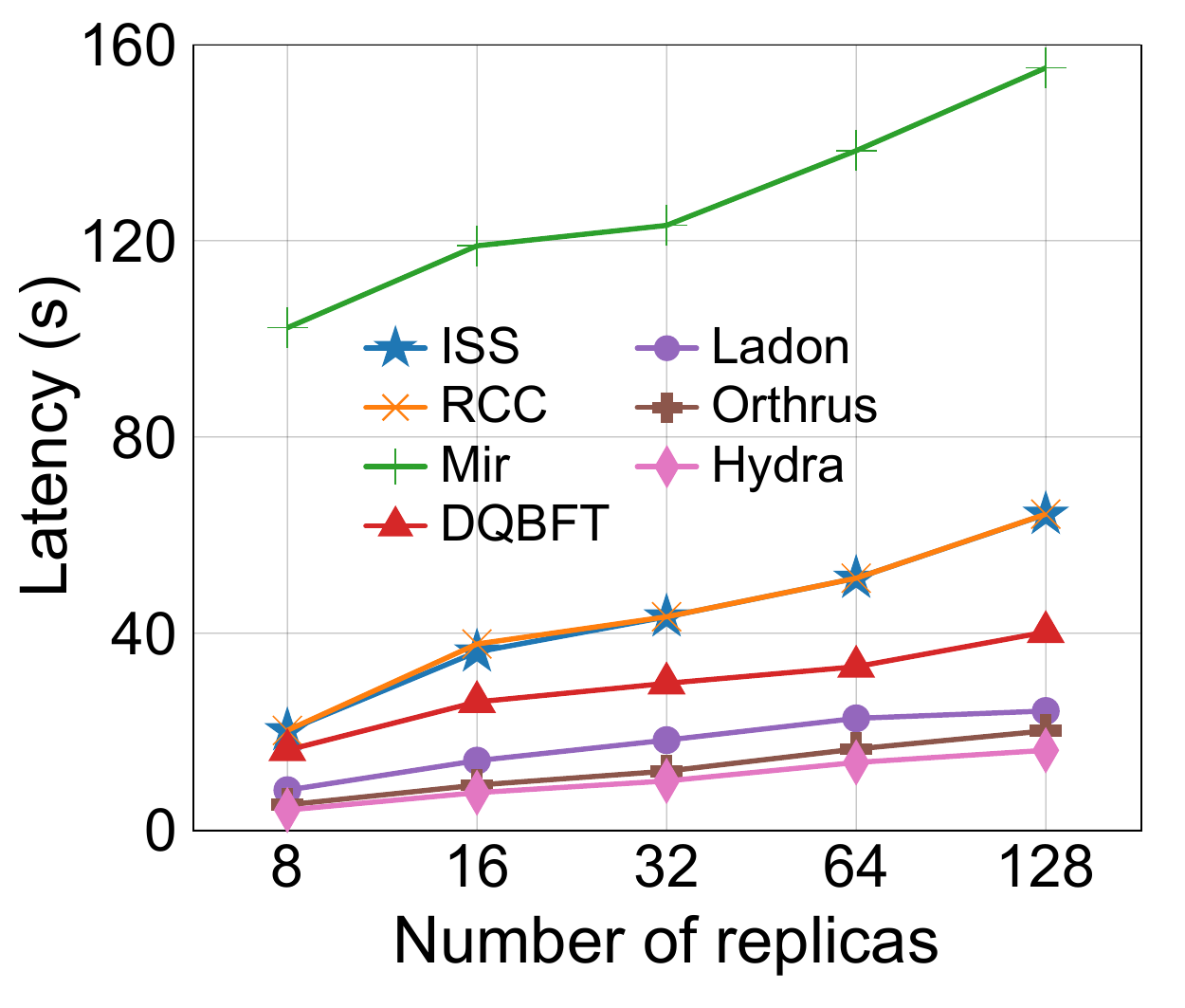}
        \caption{$\#$Straggler = 1, WAN}
        \label{fig:wan2}
    \end{subfigure}
    \hfill
    
    \begin{subfigure}[t]{0.24\textwidth}
        \centering
        \includegraphics[width=\textwidth]{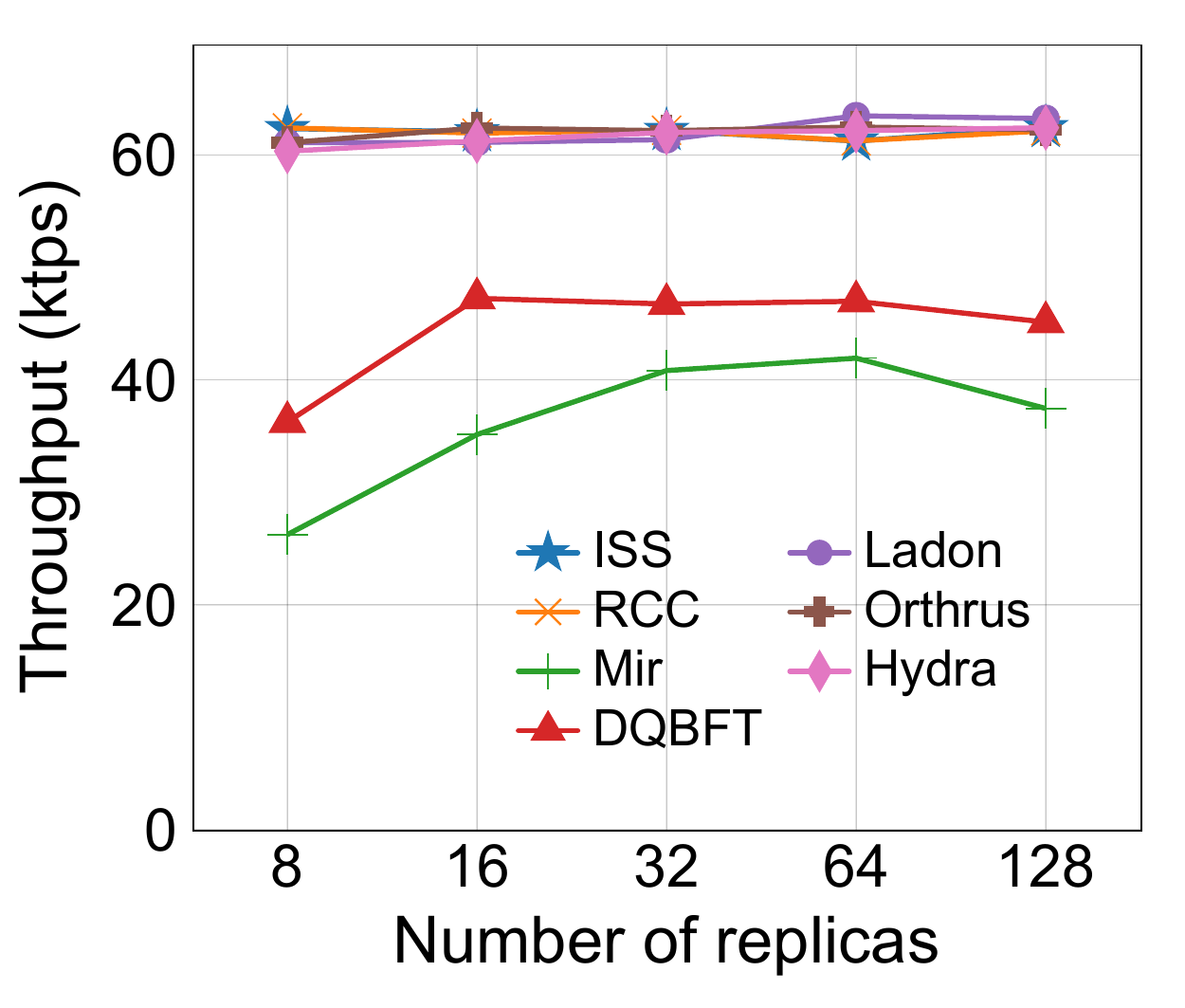}
        \caption{$\#$Stragglers = 0, WAN}
        \label{fig:wan3}
    \end{subfigure}
    \hfill
    \begin{subfigure}[t]{0.24\textwidth}
        \centering
        \includegraphics[width=\textwidth]{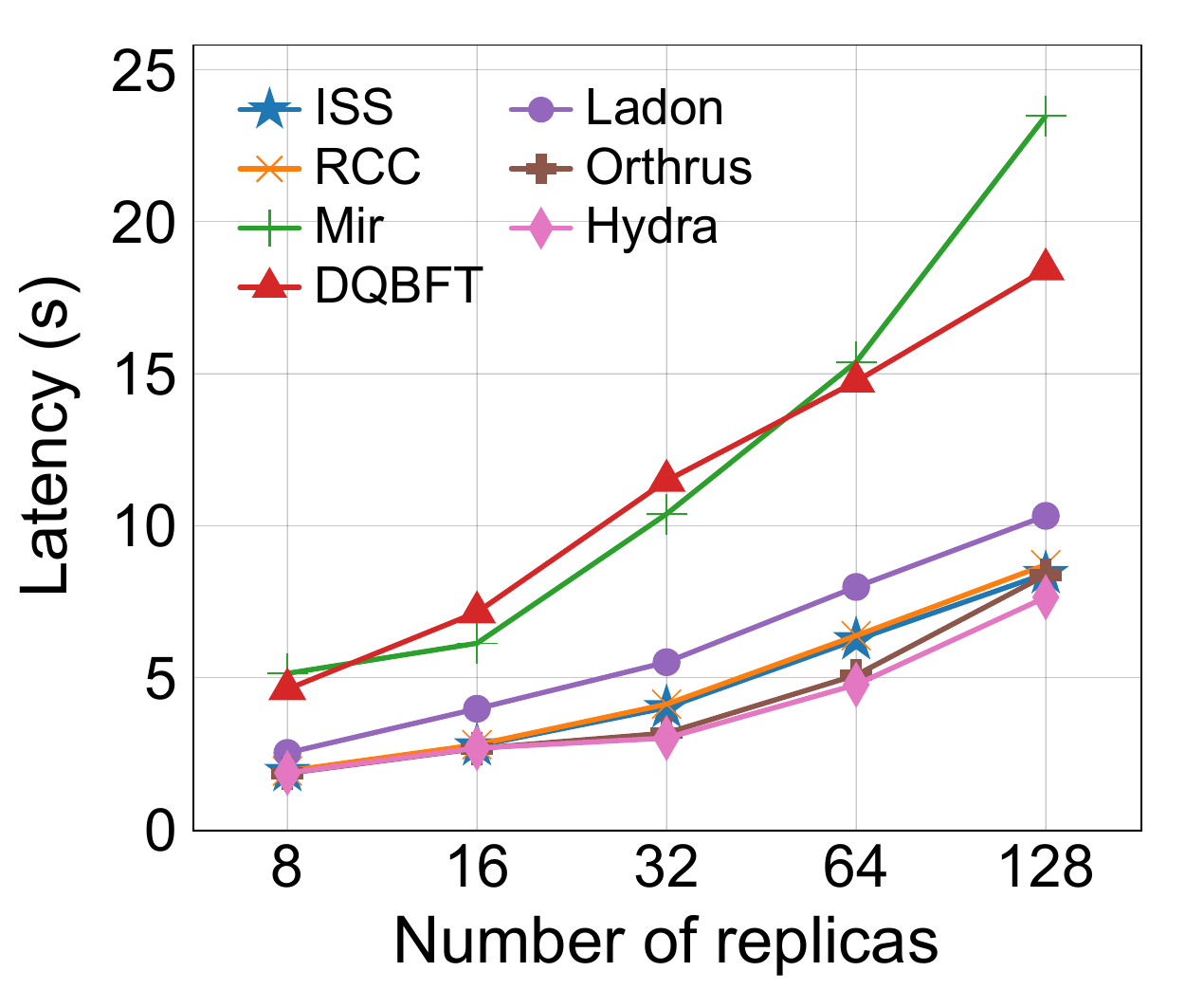}
        \caption{$\#$Stragglers = 0, WAN}
        \label{fig:wan4}
    \end{subfigure}
    \hfill\\
    
    \caption{Throughput and latency of \sysname, ISS, RCC, Mir, DQBFT, Ladon, and Orthrus in WAN.}
    \label{fig:performancewan}
\end{figure}

\begin{figure}[t]
\vspace{-4mm}
    \centering

    \begin{subfigure}[t]{0.24\textwidth}
        \centering
        \includegraphics[width=\textwidth]{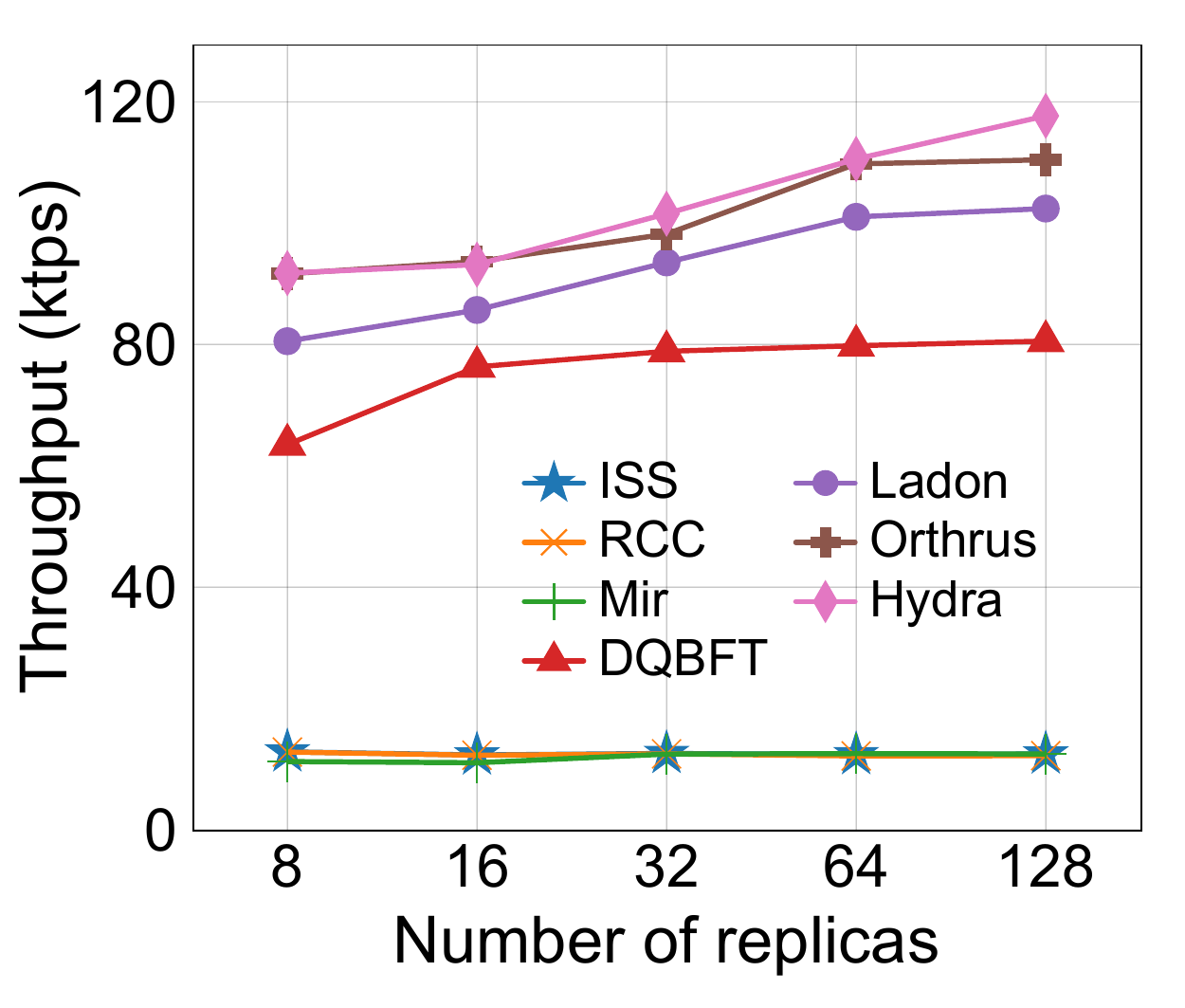}
        \caption{$\#$Straggler = 1, LAN}
        \label{fig:lan1}
    \end{subfigure}
    \hfill
    \begin{subfigure}[t]{0.24\textwidth}
        \centering
        \includegraphics[width=\textwidth]{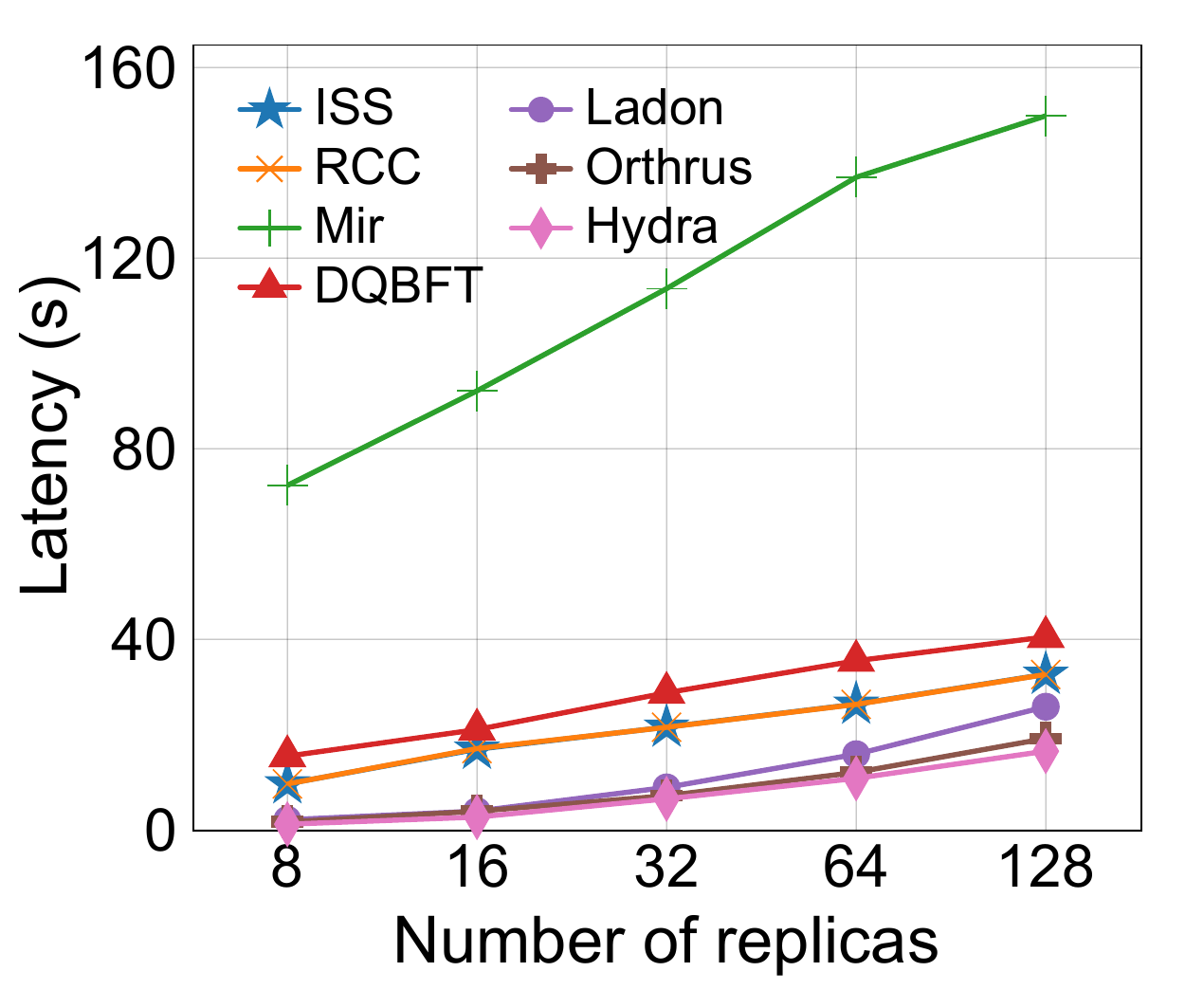}
        \caption{$\#$Straggler = 1, LAN}
        \label{fig:lan2}
    \end{subfigure}
    \hfill
    
    \begin{subfigure}[t]{0.24\textwidth}
        \centering
        \includegraphics[width=\textwidth]{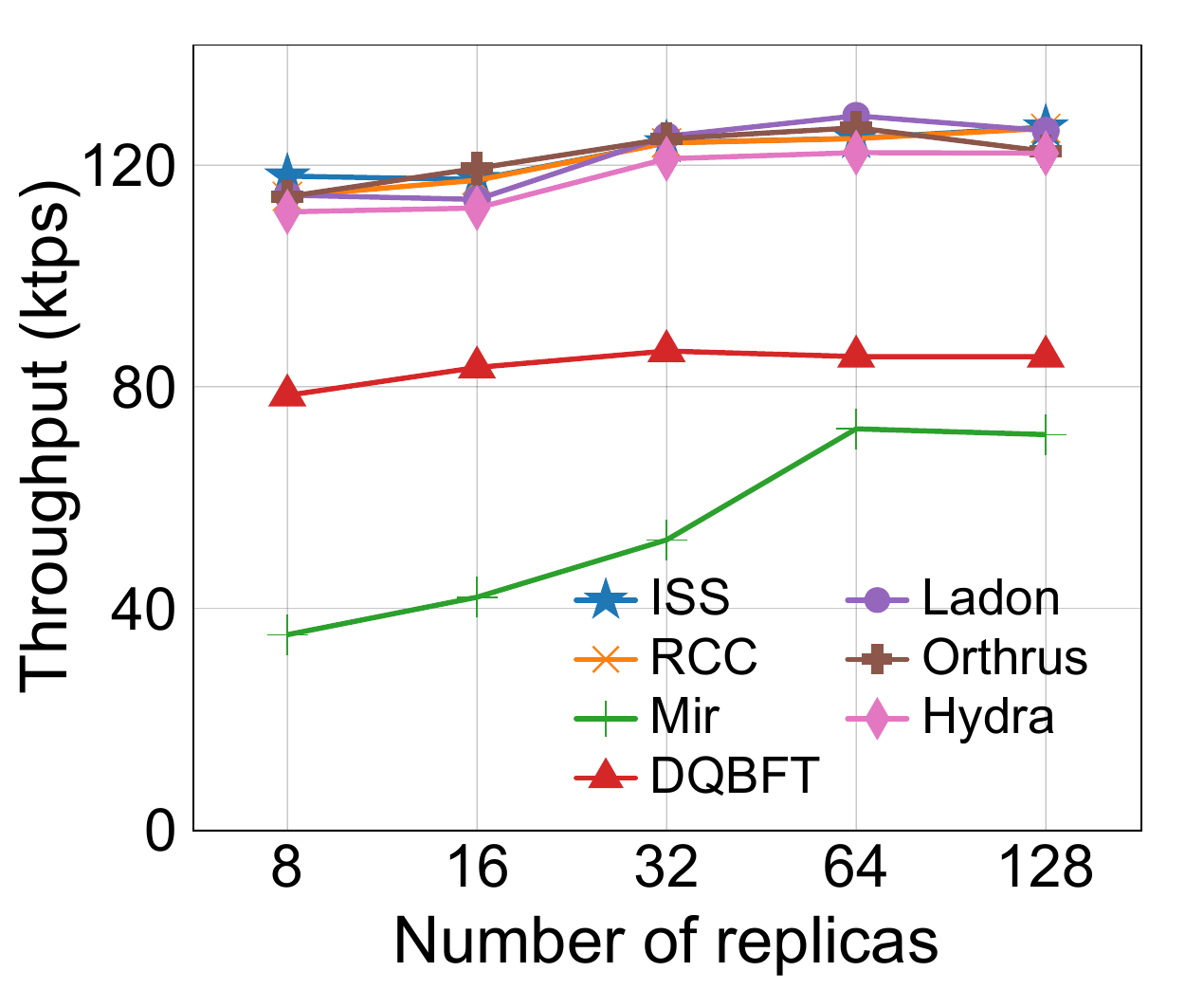}
        \caption{$\#$Stragglers = 0, LAN}
        \label{fig:lan3}
    \end{subfigure}
    \hfill
    \begin{subfigure}[t]{0.24\textwidth}
        \centering
        \includegraphics[width=\textwidth]{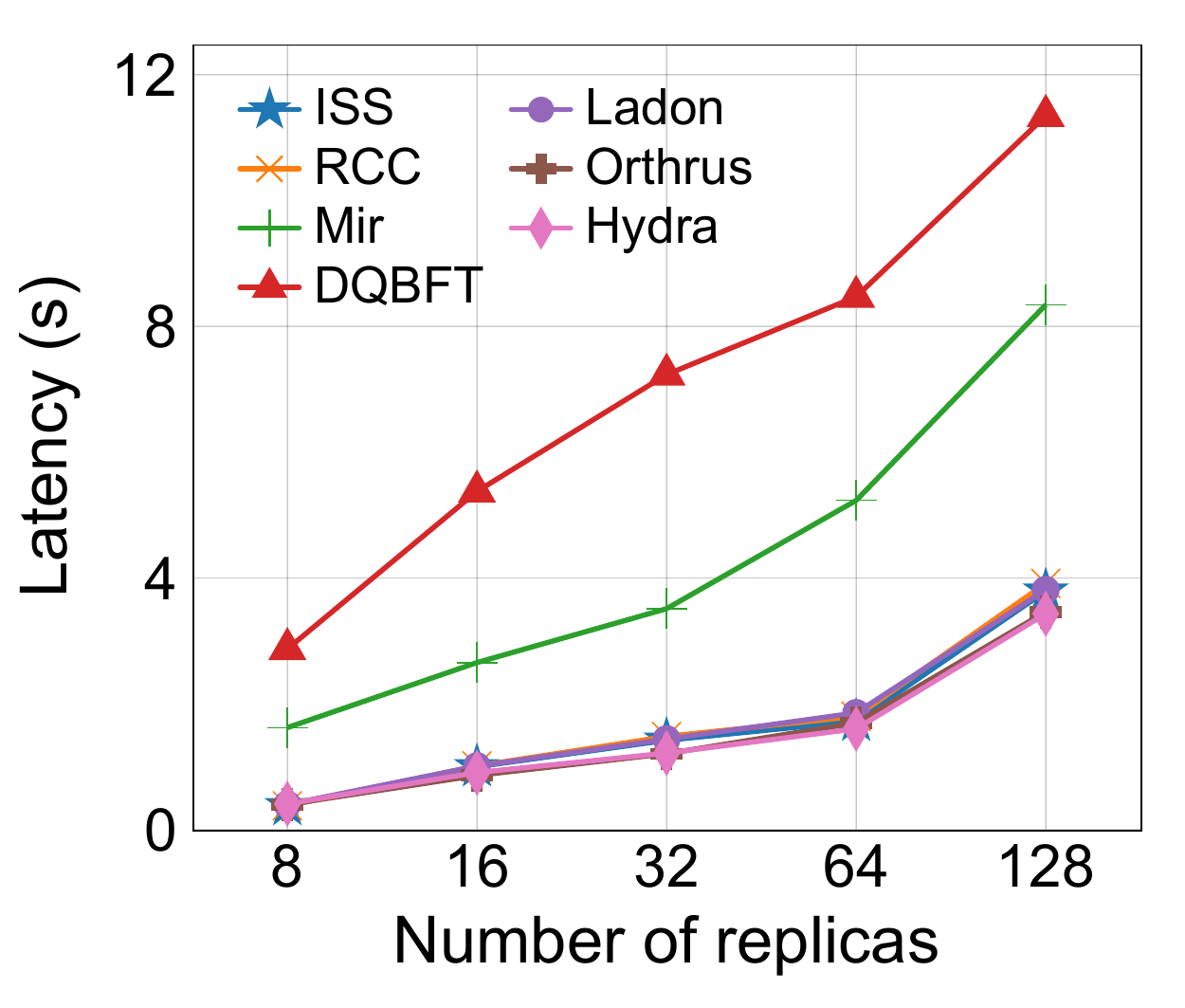}
        \caption{$\#$Stragglers = 0, LAN}
        \label{fig:lan4}
    \end{subfigure}
    \hfill\\
    
    \caption{Throughput and latency of \sysname, ISS, RCC, Mir, DQBFT, Ladon, and Orthrus in LAN.}
    \label{fig:performancelan}
\end{figure}

\bheading{Performance in LAN.} \figref{fig:performancelan} compares \sysname and other Multi-BFT protocols in LAN, with and without stragglers. The overall trends are consistent with the WAN results: \sysname maintains higher throughput and lower latency than the compared Multi-BFT protocols with one straggler. 
With 32 replicas, \sysname achieves a 7.4× throughput improvement and a 70\% latency reduction compared to the pre-determined ordering schemes. These designs remain bottlenecked by fixed global ordering. 
Besides, \sysname continues to outperform dynamic and hybrid ordering schemes by eliminating global ordering. 
When no stragglers are present, \sysname delivers throughput comparable to ISS, RCC, Ladon, and Orthrus, while sustaining consistently low latency.

\begin{figure}[t]
        \includegraphics[width=0.5\textwidth]{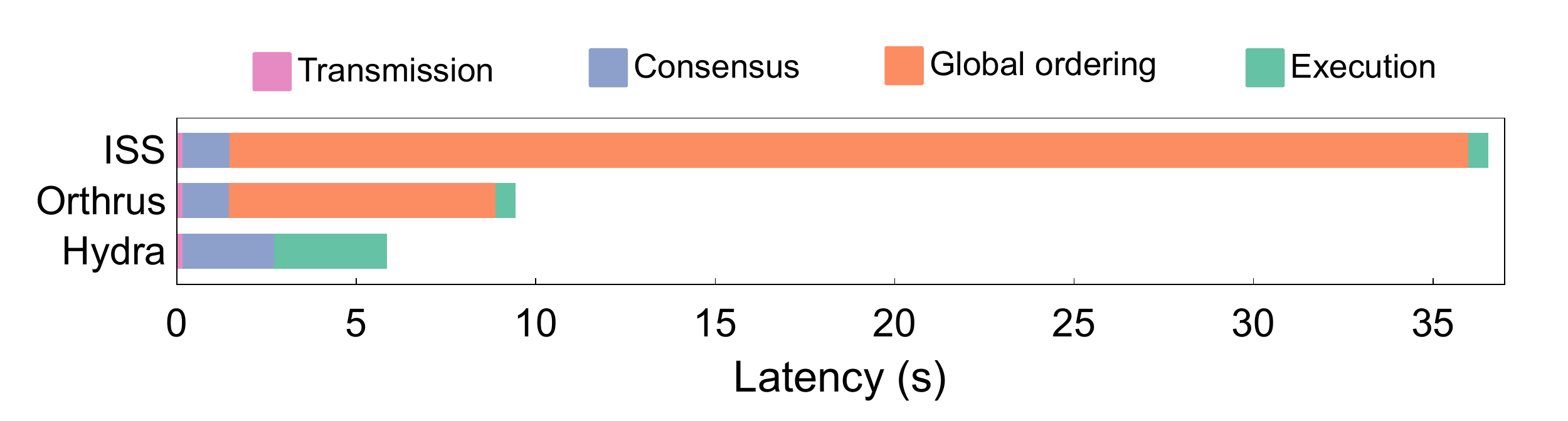}
    \caption{Breakdown of latency in ISS, Orthrus,  and \sysname.}
    \label{fig:breakdown}
\end{figure}

\bheading{Latency breakdown.} 
To better understand the source of performance differences, we present a detailed latency breakdown of ISS, Orthrus and \sysname under one straggler with 16 replicas in WAN. 
We select ISS and Orthrus because ISS represents state-of-the-art Multi-BFT protocols with pre-determined global ordering, while Orthrus illustrates dynamic and hybrid ordering mechanisms that mitigate straggler effects.
The total latency is divided into four stages: transmission, consensus, global ordering, and execution.

As shown in \figref{fig:breakdown}, the total latency of ISS is dominated by the global ordering phase, which accounts for 91.5\% of its total latency. 
Although Orthrus reduces this cost by 78.5\%, it still constitutes {69.3\%} of total latency. In contrast, \sysname eliminates global ordering entirely, resulting in significantly lower overall latency. Specifically, \sysname reduces end-to-end latency to {7.73s}, outperforming ISS and Orthrus by {79.5\%} and {27.9\%}, respectively.

Among the remaining components, the consensus latency of \sysname is slightly higher than in ISS and Orthrus. This is because we account for the time between receiving transactions and packing them into blocks as part of the consensus phase. Furthermore, since some transactions are assigned to multiple instances in \sysname, the preprocessing overhead for these transactions is slightly larger.
Notably, \sysname incurs a slightly higher execution cost due to lock acquisition and deadlock resolution. 
However, this overhead is far outweighed by the savings from removing global ordering.
Overall, these results highlight the core advantage of \sysname's design: by removing global ordering and enabling concurrent, lock-based execution, it substantially improves responsiveness without compromising consistency.

\subsection{Impact of Cross-Instance Transactions}\label{sec:hydraproportion}
\begin{figure}[t]
\vspace{-4mm}
	\centering
    \begin{subfigure}[t]{0.24\textwidth}
	\centering
        \includegraphics[width=1.05\linewidth]{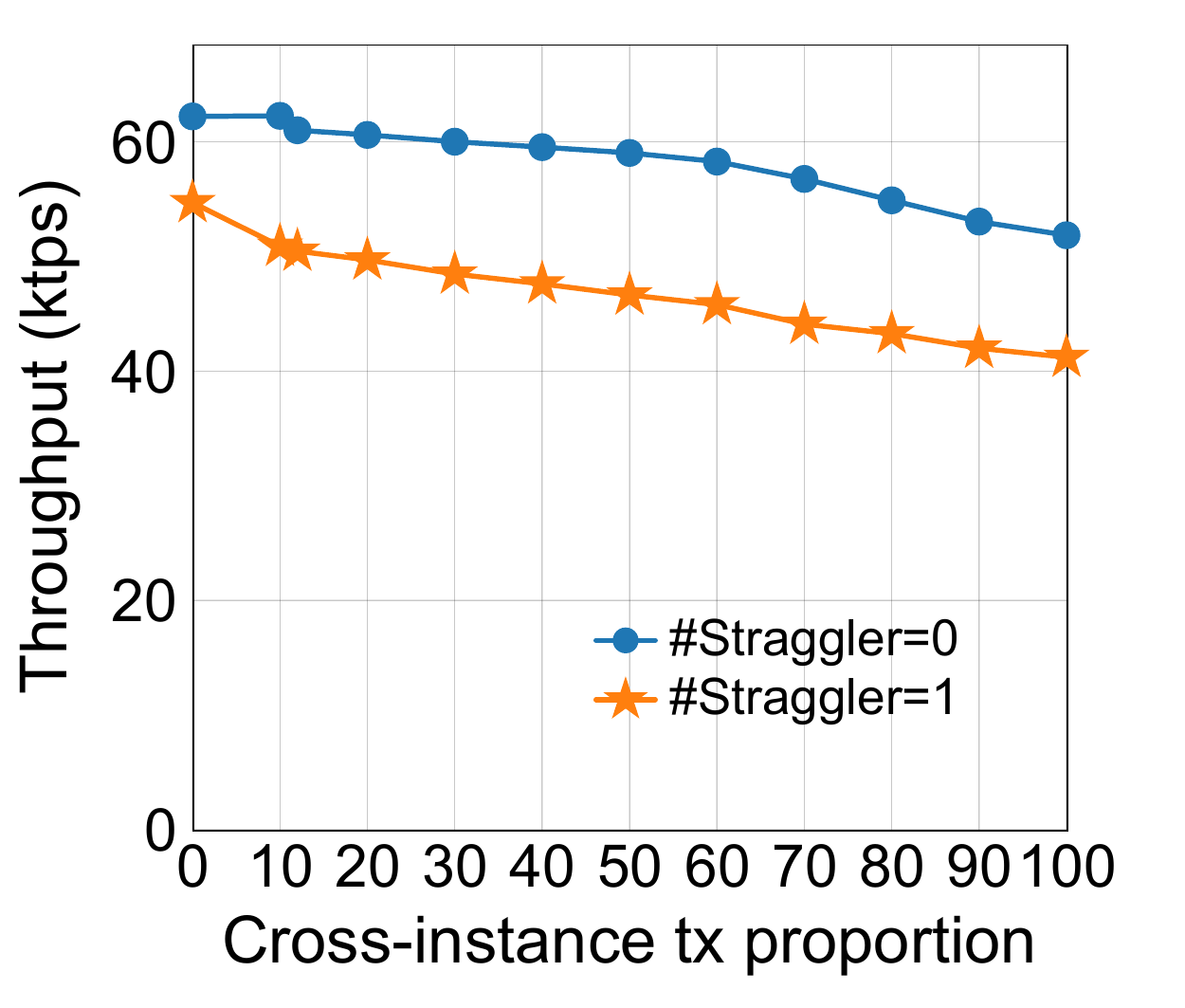}
        \caption{{Throughput}}
    \end{subfigure}
    \hfill
    \begin{subfigure}[t]{0.24\textwidth}
	\centering
        \includegraphics[width=1.05\linewidth]{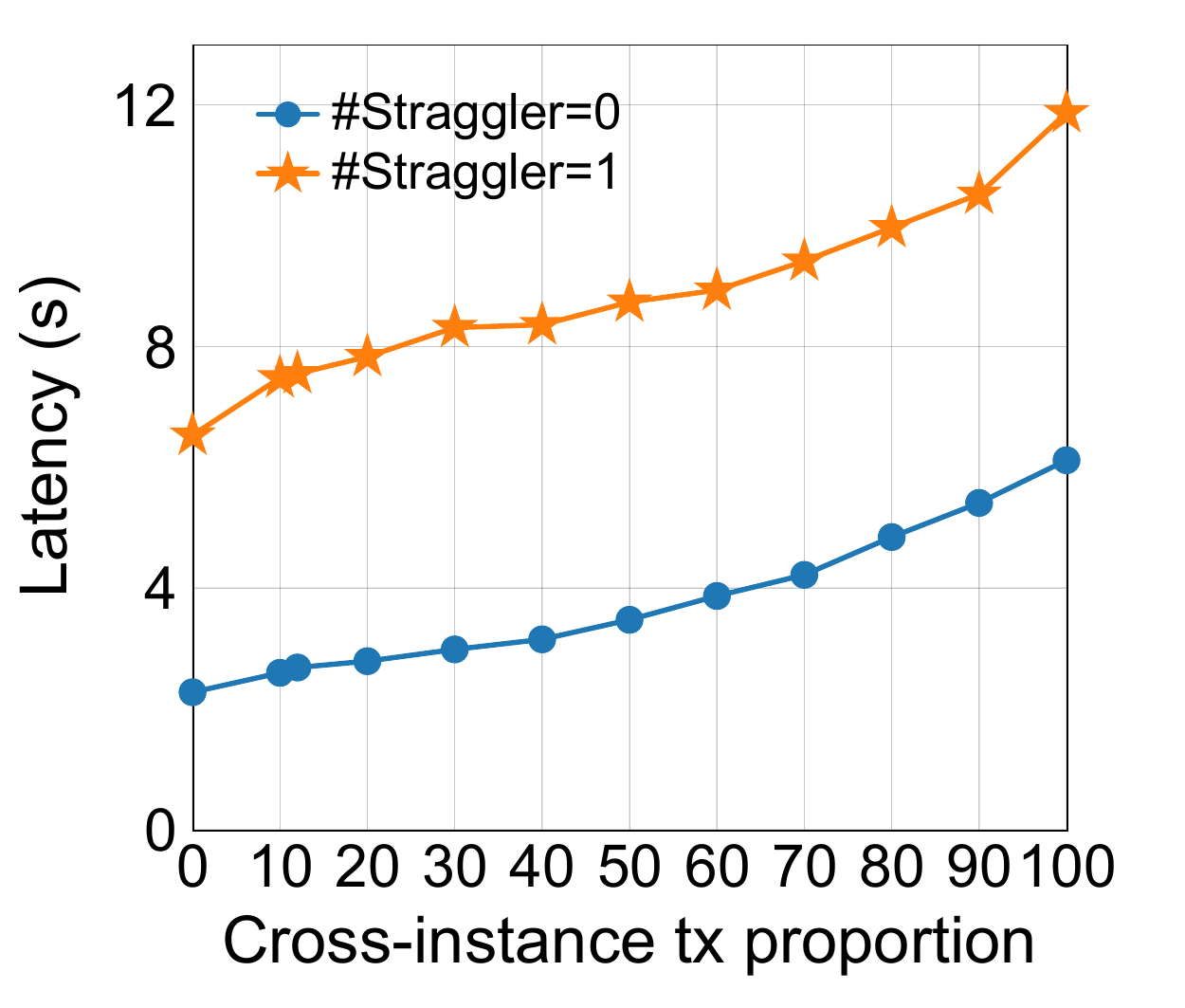}
        \caption{{Latency}}
    \end{subfigure}
    
    \caption{Throughput and latency of \sysname under different cross-instance transactions proportions in WAN.}
	\label{fig:proportionTest}
    \vspace{-4mm}
\end{figure}

\figref{fig:proportionTest} shows the throughput and latency of \sysname with different proportions of cross-instance transactions with 16 replicas in WAN. We only report WAN results here, as the LAN experiments exhibit the same overall trends.

As the proportion of cross-instance transactions increases, \sysname experiences performance degradation, indicating that coordination across instances introduces additional lock contention and higher chances of deadlock resolution.

With one straggler, when the proportion grows from 0\% to 100\%, throughput decreases by 24.6\% and latency increases by 81.7\%.
Without stragglers, the trend is similar: throughput decreases by 16\% and latency increases by 1.7×. 
In both cases, the changes are relatively gradual, showing that the cost of handling cross-instance dependencies grows moderately and does not cause severe performance degradation.

\subsection{Performance under Faults}\label{sec:hydrafaults}
\bheading{Crash faults.}
We evaluate the performance of \sysname under crash faults in a WAN setting with 16 replicas, considering three scenarios with 0, 1, and 5 faulty replicas.
Faults are injected at approximately 15s, and detected within a few seconds, triggering a view change after which performance stabilizes.

\figref{fig:detectable} shows the average throughput and latency over time. 
With one faulty replica, latency exhibits a short spike, peaking at 8.2s, and then stabilizes.
In steady state (after view change), average latency increases modestly to 2.7s (about 8.7\% higher than the pre-fault value), while throughput declines slightly by around 3-4\%.
This is because most instances continue processing normally, and only the instance led by the faulty replica temporarily stalls.

With five faulty replicas, the effect is more pronounced.
Latency briefly spikes to 12s then settles to a higher steady-state level of 2.8s, roughly 15\% higher than in the pre-fault state.
Throughput also decreases by about 75\% immediately after injection, and stabilizes at a steady-state throughput approximately 30\% lower than normal.
Even so, unaffected instances maintain stable progress and \sysname fully recovers performance once the view change completes.

Overall, these results demonstrate that failures do not cascade across instances, and \sysname remains resilient with bounded performance degradation even under multiple faults.

\begin{figure}[t]
	\centering
    \begin{subfigure}[t]{0.5\textwidth}
	\centering
        \includegraphics[width=\linewidth]{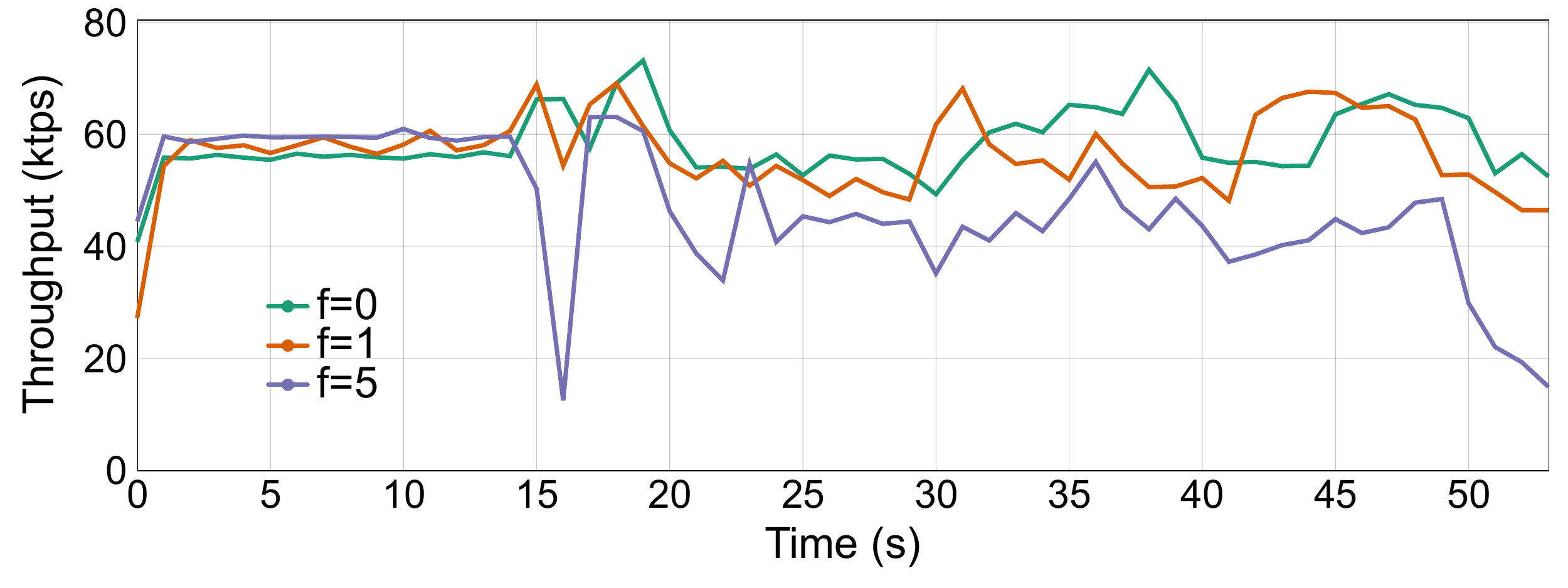}
        \caption{{Throughput average  (over 1s intervals) over time.}}
	\label{fig:detectablea}
    \end{subfigure}
    \hfill
    \begin{subfigure}[t]{0.5\textwidth}
	\centering
        \includegraphics[width=\linewidth]{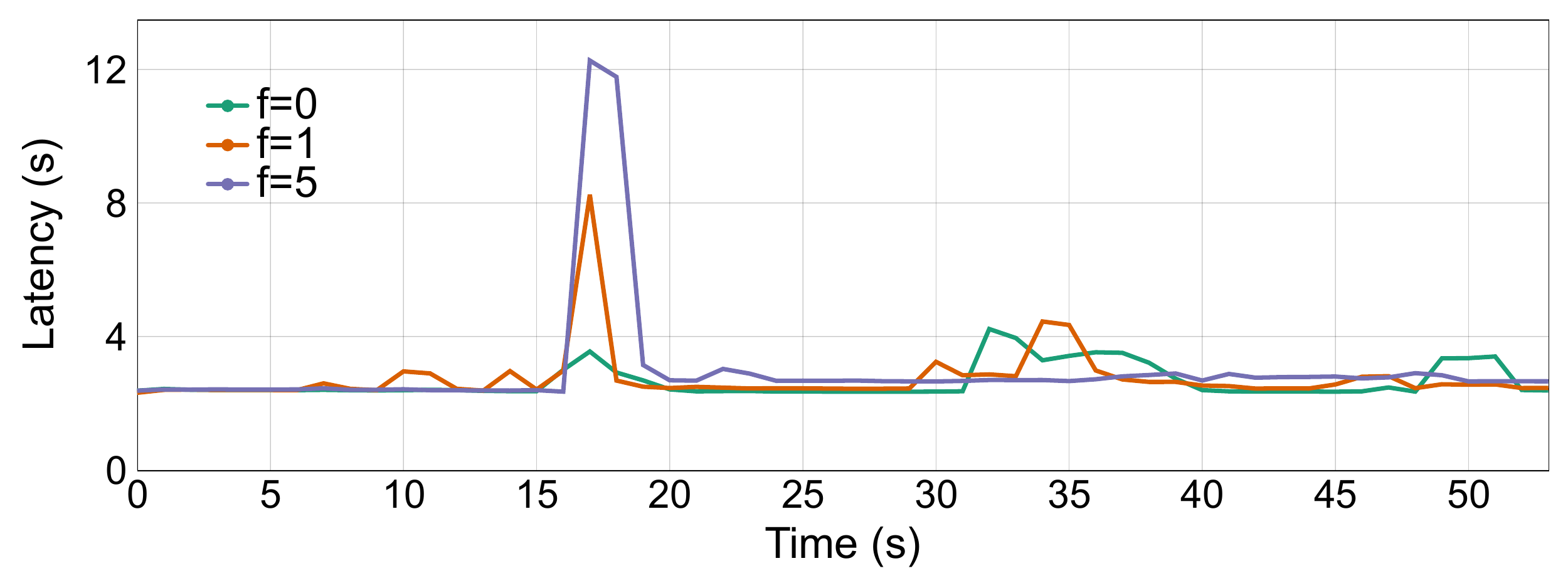}
        \caption{{Latency average  (over 1s intervals) over time.}}
        	\label{fig:detectableb}
    \end{subfigure}
    \caption{Throughput and latency average of \sysname over time with 0, 1, 5 faults. The faults occur at 15 seconds.}
	\label{fig:detectable}
\end{figure}

\begin{table}[t]
\centering
\caption{Throughput and latency under different adversarial scenarios. BL: Byzantine Leader, MC: Malicious Clients.}
\label{tab:failure}
\begin{tabular}{l c c c c}
\toprule
\textbf{Metric / Scenario} 
    & Failure-free & MC & BL & BL + MC \\
\midrule
Throughput (ktps) & 61.50 & 57.01 & 60.13 & 52.86 \\
Latency (s)        & 2.69  & 3.26  & 3.02  & 3.60 \\
\bottomrule
\end{tabular}
\end{table}

\bheading{Byzantine faults.}
Table~\ref{tab:failure} summarizes system performance under different adversarial scenarios, including high-contention workloads generated by malicious clients (MC), Byzantine leader behavior (BL), and the combination of both (BL+MC). In the BL scenario, we inject a single Byzantine leader among 16 replicas, which deliberately prolongs lock holding and proposes adversarial transaction orderings to induce cross-instance deadlocks.
In the MC scenario, malicious clients generate high-contention workloads by repeatedly submitting transactions involving the same set of 20 objects to different instances in adversarial orders, intentionally inducing cross-instance deadlocks and increasing coordination overhead. 
Each scenario captures realistic adversarial strategies that may increase lock contention or coordination delays, as discussed in Section~\ref{sec:byzantine}.

As the results show, although adversarial behaviors can reduce throughput and increase latency, the impact on Hydra is limited. 
Malicious clients submitting high-contention transactions reduce throughput by 7.3\% and increase latency by 21.2\%. Byzantine leaders attempting to manipulate locks or propose conflicting orders cause a smaller throughput reduction of 2.2\% and a latency increase of 12.3\%. When both adversarial behaviors are combined (BL+MC), throughput decreases by 14.0\% and latency increases by 33.8\%.
These results indicate that \sysname's locking and deadlock handling mechanisms effectively contain the effects of adversarial manipulation, restricting the impact to performance degradation.

\subsection{Memory Usage}\label{sec:mem}

\begin{table}[t]
\centering
\caption{Memory usage and utilization under different numbers of instances. Each instance: 8 vCPUs and 16~GB RAM.}
\label{tab:instance-memory}
\begin{tabular}{l c c c c}
\toprule
\textbf{\# Instances} & 8 & 16 & 32 & 64 \\
\midrule
\textbf{Memory Usage (MB)} & 11046 & 10648 & 10549 & 10560 \\
\textbf{Utilization} & 67.4\% & 65.0\% & 64.4\% & 64.5\% \\
\bottomrule
\end{tabular}
\end{table}

To evaluate the scalability costs of maintaining a fully replicated state, we measure the memory usage of \sysname while scaling the number of consensus instances from 8 to 64. 
Table~\ref{tab:instance-memory} reports the memory usage and utilization per instance. We observe that the memory utilization remains stable, ranging from 64.4\% to 67.4\%, despite an 8$\times$ increase in the number of instances. This indicates that \sysname exhibits near-constant memory scaling with respect to the number of instances, and that fully replicated state does not become a scalability bottleneck as the system scales horizontally.

This behavior is attributed to \sysname's epoch-based processing. At the end of each epoch, replicas perform checkpointing, truncate committed logs, and apply garbage collection to reclaim memory. As a result, the in-memory state maintained by each replica is bounded by the most recent checkpoint and does not grow unbounded over time, ensuring sustainable memory usage for long-running deployments.

\section{Discussions} \label{sec:disscuss}

\bheading{Comparison with DAG-based protocols.} 
Both Multi-BFT and DAG-based protocols aim to improve throughput by exploiting parallelism. In Multi-BFT, each block in a round references a single block from the previous round, whereas DAG-based protocols typically reference at least $2f+1$ blocks from the previous round. DAG-based systems require a global ordering to ensure consistency, which can become a performance bottleneck in the presence of slow replicas or stragglers. Moreover, DAGs may suffer from transaction duplication due to multiple block proposers each round. 

\bheading{Comparison with sharding protocols.} 
Sharding-based protocols increase parallelism by partitioning both replicas and transactions across multiple shards. However, this design introduces additional coordination overhead for cross-shard transactions, as operations spanning multiple shards require multi-shard atomicity protocols. 
\sysname’s core advantage lies in its communication-free cross-instance execution model, which fundamentally differentiates it from existing approaches. 
This point is further elaborated in the Appendix C of \cite{hydralong}.

\bheading{Distributed transaction.} 
The key of distributed transaction processing is to ensure atomicity across partitioned data, typically through coordination protocols such as 2PC~\cite{al2017chainspace,omniledger}.
While effective, these protocols incurs at least two rounds of communication and may block if the coordinator fails or if some shards are slow to respond. In contrast, \sysname operates in a fully replicated model where every replica maintains the complete object state. Cross-instance transactions are executed atomically through local locking, avoiding explicit inter-instance coordination and significantly reducing communication overhead while preserving deterministic consistency.

\bheading{Transaction partition strategies.}
\sysname does not assume a specific transaction partitioning algorithm, but instead is designed to operate correctly and efficiently under a wide range of existing partitioning strategies.
For example, a simple and widely adopted approach assigns objects to instances by hashing their keys, \ie, ${insIndex} = \mathsf{Hash}({key}) \bmod m$.
A transaction is then routed to the instances responsible for the objects it accesses.
This simple strategy 
achieves good load balance in expectation, but it may increase the frequency of cross-instance transactions when objects that frequently interact are mapped to different instances.

To mitigate this effect, \sysname can naturally benefit from prior work such as TxAllo~\cite{zhang2023txallo}, which reduces the frequency of cross-instance transactions while maintaining a balanced load across instances.
By co-locating frequently interacting objects when possible, such approaches lower cross-instance coordination overhead and deadlock probability, thereby improving overall throughput.

Importantly, these partitioning strategies are orthogonal to the core design of \sysname.
The system’s correctness guarantees hold under any deterministic partitioning scheme agreed upon by all replicas.

\section{Related Work}\label{sec:related}

\bheading{Multi-BFT consensus.}
The foundation of Byzantine fault-tolerant consensus was established by Castro and Liskov in PBFT~\cite{pbft1999}, which inspired a long line of practical protocols such as Zyzzyva~\cite{kotla2007zyzzyva}, Tendermint~\cite{Buchman2016TendermintBF}, and HotStuff~\cite{hotstuff}.
These leader-based protocols streamline the process of reaching agreement but inevitably concentrate decision power on a single replica, resulting in a leader bottleneck that constrains throughput and scalability~\cite{alqahtani2021bottlenecks,charapko2021pigpaxos,gai2023scaling,kang2025hotstuff,gupta2023dissecting}.

To alleviate this limitation, Multi-BFT consensus architectures were introduced, allowing replicas to run several consensus instances concurrently and thereby removing the dependence on a single leader~\cite{stathakopoulou2022state, gupta2021rcc, MIR-BFT, Ladon2025,kang2024spotless}.
Among early designs, Mir-BFT~\cite{MIR-BFT} pioneered this direction by executing multiple instances in parallel and establishing a global order of blocks based on pre-assigned indices.
However, because Mir-BFT triggers an epoch change whenever any instance leader behaves incorrectly, the system remains vulnerable to Byzantine disruptions.
ISS~\cite{stathakopoulou2022state} refined this model by introducing no-op deliveries to mitigate unnecessary epoch changes, while RCC~\cite{gupta2021rcc} similarly relies on a static global order of outputs.
In both cases, a slow or faulty instance can delay epoch progression for all replicas, creating performance bottlenecks.

Other variants attempt to optimize the global ordering process.
For example, DQBFT~\cite{dqbft} dedicates a separate BFT instance to globally order outputs from other instances, simplifying coordination but making the system more prone to targeted attacks~\cite{estrada2006network}.
Ladon~\cite{Ladon2025} introduces monotonic ranking to dynamically determine cross-instance ordering, mitigating straggler-induced delays. Building upon this, Orthrus~\cite{Orthrus2025} further introduces a fast path for independent transactions, enabling them to bypass the global ordering phase, while still maintaining global consistency for dependent ones. 
Although these techniques alleviate synchronization overheads, the global ordering phase still constitutes a large fraction of overall system latency.

Unlike the systems above, TELL~\cite{tong2024tell} shifts optimization efforts to the execution layer instead of the consensus protocol itself.
It employs a State Hash Table (SHT) to track read/write dependencies, enabling concurrent execution both across instances and within blocks.
Conflicts are handled by selective re-execution and periodic merging of instance states at the epoch boundary.
While this strategy reduces execution wait time, its improvement in end-to-end latency remains limited due to the overhead of reprocessing conflicting transactions.

\bheading{Partial ordering design.}
A complementary research direction focuses on partially ordered execution to relax the global ordering requirement, particularly within payment-oriented systems.
CryptoConcurrency~\cite{tonkikh2023cryptoconcurrency} dynamically detects overspending by checking account balances, enabling concurrent transactions without full consensus.
Astro~\cite{collins2020online} maintains per-client logs to prevent double spending, while ABC~\cite{sliwinski2021asynchronous} allows validators to process transactions in parallel without global coordination, thereby enhancing efficiency.
FastPay~\cite{baudet2020fastpay} leverages payment semantics to minimize shared state between accounts, facilitating asynchronous execution with high concurrency.
Flash~\cite{lewis2023flash} bypasses reliable broadcast through a DAG-based, partially ordered structure.
A non-sequential model for monetary transfers is proposed in~\cite{auvolat2020money}, based on reliable broadcast abstraction.
Pastro~\cite{kuznetsov2023permissionless} further defines a partially ordered transaction set that determines active participants and stake distribution, offering flexibility for applications that do not depend on a total order.
While these approaches achieve efficient execution for payment systems, they cannot generally support complex smart contract semantics. 

\bheading{Hybrid ordering design.}
Recent advances in blockchain systems have explored hybrid ordering mechanisms that selectively apply global ordering only when necessary~\cite{blackshear2023sui,babel2025mysticeti}.
Sui Lutris~\cite{blackshear2023sui} introduces a dual-path model, where single-owner transactions follow a lightweight fast path, while shared-object transactions are ordered through a consensus path. This design reduces latency for independent operations but depends on client-side orchestration, which requires gathering additional certificate signatures and performing periodic checkpointing to ensure finality.
Mysticeti~\cite{babel2025mysticeti} generalizes the same concept using a DAG-based consensus protocol, offering a tighter integration between the fast and consensus paths. However, Mysticeti still separates execution and finality: locally executed transactions may later be reverted if not sufficiently confirmed at epoch boundaries, complicating persistence across epochs.
Orthrus~\cite{Orthrus2025} extends this hybrid ordering paradigm into a Multi-BFT framework while retaining a standard BFT core. It introduces a fast path where conflict-free transactions can be finalized immediately upon instance commitment, thus eliminating deferred confirmation and preventing rollback between epochs.
Thunderbolt~\cite{chen2024thunderbolt} introduces a hybrid ordering mechanism that deterministically orders cross-shard transactions while preserving partial concurrency for single-shard execution, achieving both high scalability and consistency across shards.

\section{Conclusion}\label{sec:conclusion}
We presented \sysname, a new Multi-BFT consensus architecture that eliminates global ordering while preserving safety and liveness. By adopting an object-centric execution model and enforcing only per-object ordering, \sysname decouples the progress of parallel BFT instances and unlocks significantly higher concurrency. Our locking mechanism and deterministic deadlock handling further ensure that cross-instance execution remains correct without additional coordination rounds.
We implemented \sysname and evaluated it under both WAN and LAN environments. The results demonstrate that \sysname sustains high throughput and low latency even in the presence of stragglers, outperforming state-of-the-art Multi-BFT protocols by up to 9× in throughput. 

\clearpage
\section{AI-Generated Content Acknowledgment}
Portions of this paper were assisted by the use of the GPT-5 model from OpenAI. Specifically, AI assistance was used for language refinement. All conceptual contributions, algorithmic designs, experimental methodologies, and scientific claims are solely made by the authors. The authors reviewed and verified the correctness of all AI-assisted content.
\normalem
\bibliographystyle{unsrt}
\bibliography{bib}

\appendices

\section{Correctness Analysis} 
We prove the safety and liveness properties of \sysname. 
We first show that all honest replicas maintain consistent decisions on transaction aborts during deadlock detection and resolution.
We then prove that any two honest replicas that reach the same delivery frontier apply identical updates to every object, thereby guaranteeing state consistency and overall system safety.
Finally, we demonstrate that every transaction broadcast by a correct client will eventually either be executed or aborted, thereby ensuring liveness.

\begin{lemma}\label{lemma:deadlockgroup}
All honest replicas will compute the same deadlock group \( D \) when executing $\mathsf{expandDeadlockGroup}$(\( tx \)).
\end{lemma}

\begin{proof}
We prove this by induction on the recursive expansion process of the deadlock group.

\bheading{Notation.}
Let \( I \) denote the set of instances assigned to \( tx \).  
Let \( plog[i] \) be the instance log at instance \( i \), and \( prefix[i](tx) \) denote the set of transactions that appear before \( tx \) in \( plog[i] \).  
Let \( \prec_i \) denote the order of transactions in \( plog[i] \). We say that \( tx' \) and \( tx \) form a deadlock if there exist \( i, j \in I \) such that \( tx' \prec_i tx \) and \( tx \prec_j tx' \).

\bheading{Base Case.}
Since \( tx \) has been delivered by all instances in \( I \), by the Agreement property of the underlying SB protocol, all honest replicas share the same prefix \( prefix[i](tx) \) for each \( i \in I \).  
This ensures that any transaction \( tx' \) occurring before \( tx \) in any \( plog[i] \) is consistently visible to all replicas.

\bheading{Step 1: Identifying Immediate Conflicts.}
Suppose \( tx' \) forms a deadlock with \( tx \), \ie, \( tx' \prec_i tx \) for some instance \( i \) and \( tx \prec_j tx' \) for some instance \( j \).  
Then \( tx' \in prefix[i](tx) \), and hence is visible to all honest replicas.  
Moreover, since \( tx' \) is part of a deadlock cycle, it cannot have completed execution (as completed transactions release all locks and are no longer pending).  
Therefore, \( tx' \) will be added to the deadlock group by all honest replicas.

\bheading{Step 2: Inductive Expansion.}
Let \( D_n \) be the deadlock group after \( n \) expansion steps, assumed to be identical at all honest replicas.  
In the \( (n+1) \)-th step, for each \( t \in D_n \), the procedure $\mathsf{findOrderingConflicts}$(\( t \)) scans the logs to identify transactions \( t' \) that conflict with \( t \) based on relative ordering across instances.  
Since all logs are consistent at the relevant positions (due to prior delivery), all honest replicas will observe the same conflicting transactions \( t' \) and add them to the group.  
Thus, \( D_{n+1} \) remains identical across replicas.

\bheading{Conclusion.}
As the expansion process is monotonic and bounded over a finite set of transactions within the epoch, it terminates after a finite number of steps.  
All honest replicas will therefore compute the same final deadlock group \( D\) when invoking $\mathsf{expandDeadlockGroup}$(\( tx \)).  
\end{proof}

\begin{lemma}\label{lemma:abort}
If an honest replica aborts a transaction \( tx \), then all honest replicas will abort \( tx \).
\end{lemma}

\begin{proof}
A transaction \( tx \) can be aborted for one of two reasons:
(1) it is not confirmed (\ie, delivered by all of its target instances) before the end of an epoch, or
(2) it is involved in a deadlock cycle and is selected as a victim.

\bheading{Case 1: Unconfirmed.}
Let \( I \) be the set of instances \( tx \) is assigned to.  
By the {Agreement} property of the SB protocol, if \( tx \notin plog[i] \) at one honest replica for some \( i \in I \), then it is not present in \( plog[i] \) at any honest replica.  
Hence, at the end of the epoch, all honest replicas will observe that \( tx \) is not fully delivered and will independently abort \( tx \).  
This decision is thus consistent.

\bheading{Case 2: Deadlock Resolution.}
In this case, \( tx \) is aborted due to a deadlock. By Lemma~\ref{lemma:deadlockgroup}, all honest replicas will compute the same deadlock group \( D \) by invoking $\mathsf{expandDeadlockGroup}$(\( tx \)).  
All honest replicas apply the same deterministic victim selection rule (\eg, selecting the lowest-indexed transactions) in $\mathsf{selectVictims}$(\( D \)).  Since both \( D \) and the selection rule are consistent across replicas, the resulting victim set is identical at all honest replicas.
Therefore, if \( tx \in \mathsf{selectVictims}(D) \) at one honest replica, then all replicas will also abort \( tx \).

\bheading{Conclusion.}
In both cases, if one honest replica aborts \( tx \), all honest replicas will make the same decision.
\end{proof}

\begin{theorem}[Safety]    
If two honest replicas reach the same state \( S \), they must have identical values for every key in \( S \). 
\end{theorem}

\begin{proof}
We first clarify the inevitability of honest replicas reaching the same system state $S$. During an epoch, honest replicas may temporarily deliver different numbers of blocks from different SB instances due to asynchrony. However, at the end of an epoch, all honest replicas converge to the same vector of delivered maxima. Therefore, if two honest replicas reach the same state $S$, this is not incidental but guaranteed to eventually occur as a result of the epoch mechanism.

Let two honest replicas \( R_1 \) and \( R_2 \) reach the same system state \( S = (sn_0, sn_1, \dots, sn_{m-1}) \), where each \( sn_i \) denotes the maximum sequence number delivered from instance \( i \).  
By the definition of \( S \), both replicas have delivered the exact same sequence of blocks from all SB instances.
Each transaction \( tx \) in these blocks is independently processed by the replicas to determine whether it should be executed or aborted.  
By Lemma~\ref{lemma:abort}, all honest replicas make consistent abort decisions for any given \( tx \).  
Thus, all honest replicas execute the same set of transactions.
For each executed transaction \( tx \), its internal DAG determines a topological order of operations, which is interpreted identically by all replicas.  Moreover, lock acquisition and execution follow the same rules, ensuring that the execution is deterministic.  
Therefore, both \( R_1 \) and \( R_2 \) apply the same updates to each key in the same order.  
It follows that the final state reached at \( S \) is identical across \( R_1 \) and \( R_2 \), and in particular, all keys have the same values.
\end{proof}

\begin{lemma} \label{lemma:confirmed} 
 If a correct client broadcasts a transaction \( tx \),  all honest replicas will eventually confirm \( tx \).
\end{lemma}

\begin{proof}
For any transaction \( tx \) sent by a correct client, the transaction is sent to at least $f + 1$ replicas, ensuring that at least one honest replica receives it and propagates it to all honest replicas.
Each instance uses an SB protocol with guaranteed termination. Therefore, \( tx \) will eventually be proposed by an honest leader and delivered in all assigned instances.
Once all required instances have delivered \( tx \) within an epoch, it will be confirmed. If \( tx \) fails to be confirmed within an epoch (\eg, due to missing deliveries from some instances), the protocol reinserts it into the buckets. 
Reinjected transactions are placed earlier in the next epoch's processing order and have a great opportunity to be proposed and delivered. If an attempt still fails, the transaction is again reinserted and retried. This repeated reinsertion/ retry process prevents \( tx \) from remaining unconfirmed indefinitely.
Eventually, \( tx \) is delivered in all required instances, and thus confirmed by all honest replicas.
\end{proof}

\begin{lemma} \label{lemma:executed} 
If a transaction \( tx \) is confirmed, it will eventually be executed.
\end{lemma}

\begin{proof}
Once a transaction \( tx \) is confirmed, it will be scheduled for execution by all honest replicas.
If \( tx \) is not involved in any deadlock, it will successfully acquire all required object locks and be executed according to its internal DAG dependencies.

If \( tx \) is involved in a deadlock, the system invokes a deadlock detection and resolution mechanism (Algorithm~\ref{algorithm:deadlock}) that ensures eventual termination of all deadlocks.  
This is achieved by identifying consistent deadlock groups and deterministically selecting victim transactions to abort, thereby breaking the cycle.
Note that each epoch contains only a finite number of transactions, and deadlock expansion is performed solely among pending transactions within the current epoch.  
Therefore, the deadlock group is guaranteed to stabilize in a finite number of steps, ensuring that the resolution process always terminates.

Once the deadlock is resolved, each transaction in the group will either acquire the necessary locks and proceed to execution, or be selected as a victim and aborted.
If $tx$ is aborted, the transaction is reinserted into the buckets according to our design. A reinjected transaction must eventually be confirmed again, because any unconfirmed transaction will continue to be rescheduled for processing. Once confirmed, the system will retry its execution.
Since aborted transactions within a deadlock group are reinserted in a deterministic order, the same dependency cycle cannot persist indefinitely; thus, these transactions will not remain in deadlock forever and will eventually be executed successfully.
Hence, every confirmed transaction is guaranteed to eventually be executed.
\end{proof}

\begin{theorem}[Liveness]
If a correct client broadcasts a transaction \( tx \), then the client will eventually receive a response.
\end{theorem}

\begin{proof}
By Lemma~\ref{lemma:confirmed} and~\ref{lemma:executed}, if a correct client broadcasts a transaction \( tx \), then all honest replicas will eventually be executed. 
The execution result will be returned to the client.
Therefore, every transaction broadcast by a correct client is guaranteed to eventually trigger a response.
\end{proof}

\section{Additional Experiments}
\begin{figure}[t]
    \vspace{-4mm}
	\centering
    \begin{subfigure}[t]{0.24\textwidth}
	\centering
        \includegraphics[width=1.05\linewidth]{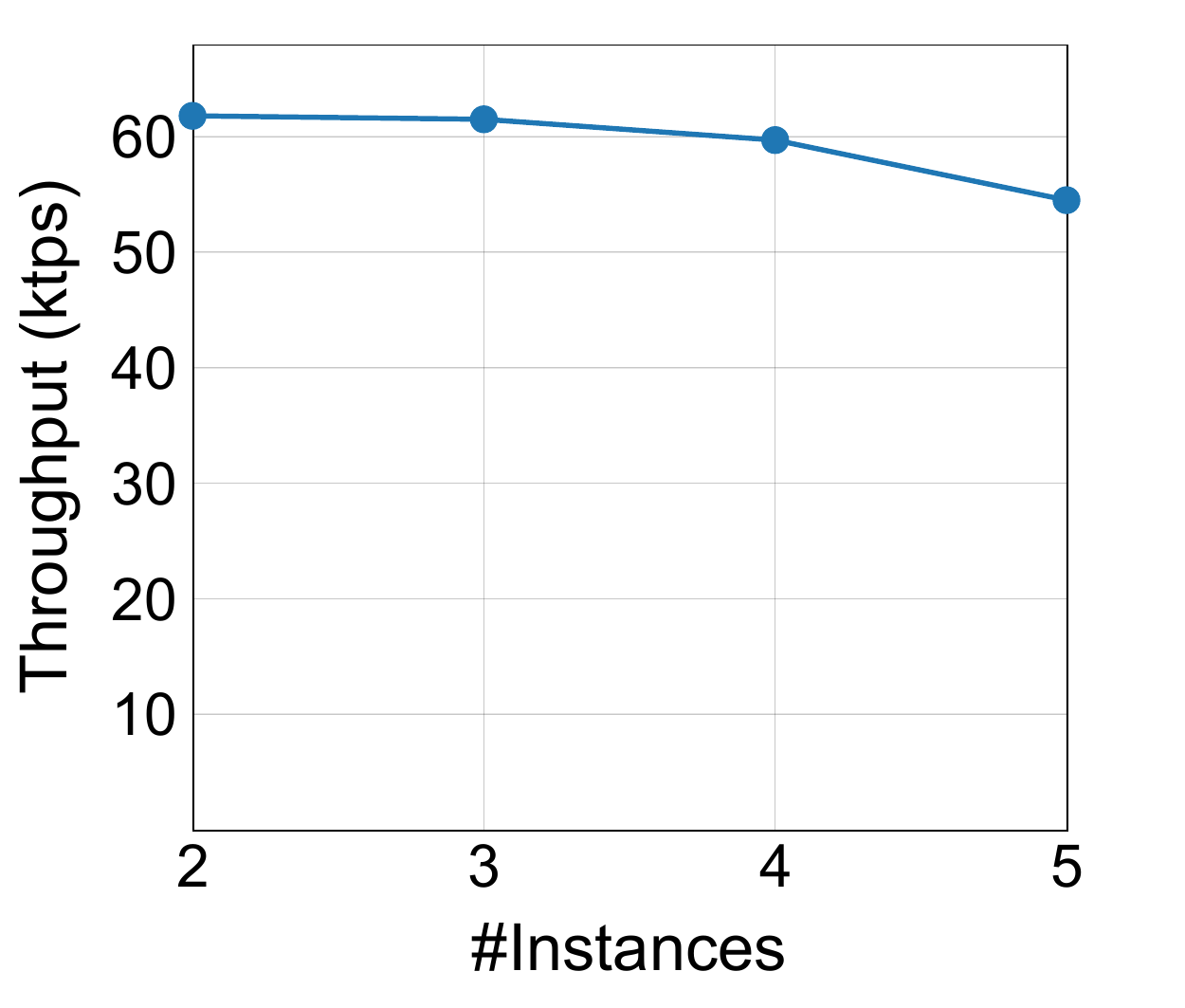}
        \caption{{Throughput}}
    \end{subfigure}
    \hfill
    \begin{subfigure}[t]{0.24\textwidth}
	\centering
        \includegraphics[width=1.05\linewidth]{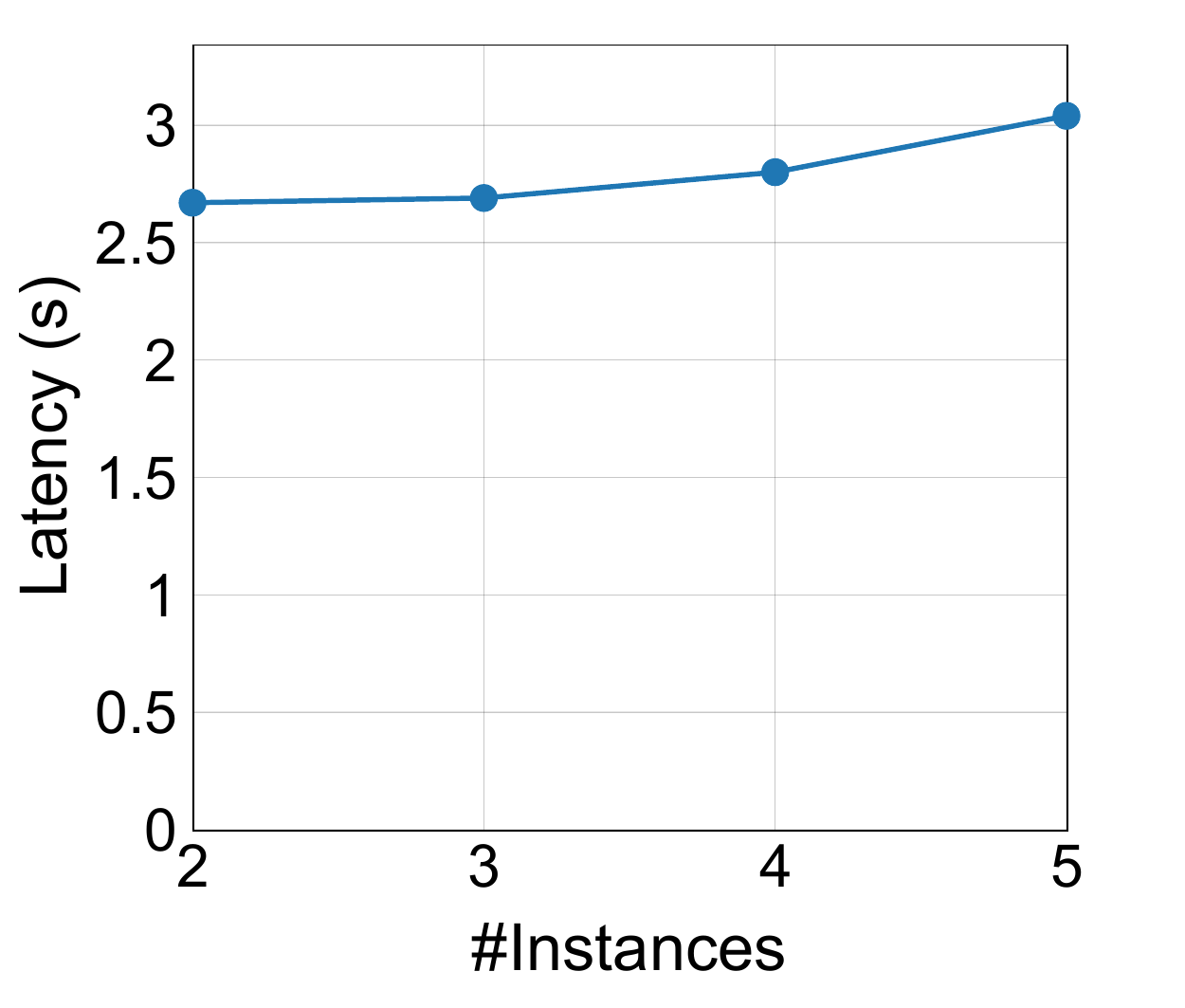}
        \caption{{Latency}}
    \end{subfigure}
    
    \caption{Throughput and latency of \sysname\ with increasing cross-instance complexity.}
	\label{fig:crossinstances}
    \vspace{-4mm}
\end{figure}

\bheading{Cross-instance transactions spanning more instances.} 
We evaluate scenarios in which each cross-instance transaction spans 2--5 instances, representing increasingly complex coordination patterns.

As shown in \figref{fig:crossinstances}, transactions spanning more instances lead to a moderate reduction in throughput and an increase in latency compared to transactions involving fewer instances, but the degradation remains bounded.
Specifically, as the number of involved instances increases from 2 to 5, throughput decreases approximately 11.8\%.
Over the same range, latency increases about 13.9\%.

These results indicate that Hydra can efficiently support cross-instance transactions involving a larger number of instances, and that increasing cross-instance transaction complexity does not introduce significant additional performance degradation.

\section{Comparison with Sharding}\label{app:sharding}
This appendix clarifies the relationship between \sysname and sharding designs. 

\subsection{Cross-Shard Transaction Handling}
Existing sharding designs fundamentally rely on coordination among the involved shards to guarantee atomicity and consistency of cross‑shard transactions. We next discuss several representative works along these directions.

\bheading{Client‑driven two‑phase commit (2PC).} 
This class originates from systems such as OmniLedger~\cite{omniledger}. In these designs, the client orchestrates a two‑phase commit across shards. In the prepare phase, the client sends the transaction to all input shards; after each shard verifies the transaction and locks the necessary state, it returns a signed proof. In the commit phase, the client aggregates these proofs and sends a combined proof to the output shards to finalize the transaction. 

\bheading{Committee‑driven two‑phase commit.}
In this category, shards coordinate among themselves without relying on clients for orchestration~\cite{al2017chainspace, hellings2021byshard,liu2025realizing}. For example, in Chainspace, where committees of replicas in each shard perform BFT consensus locally and then exchange signed messages among committees to drive prepare and commit decisions. Although this removes client dependency, it still requires explicit cross‑shard communication and multiple coordination rounds.

\bheading{Intersection replicas.}
Some protocols introduce structural overlap in the validator sets of shards. Approaches such as Pyramid~\cite{hong2021pyramid} select replicas that participate in multiple shards simultaneously, enabling them to act as natural bridges for cross‑shard transactions. These intersection replicas can validate and confirm cross‑shard state changes within a single consensus epoch, reducing the need for explicit message exchange across shards. This technique trades off increased validator complexity and careful committee assignment for lower coordination overhead.

\bheading{Deterministic ordering.} 
Several recent high‑performance blockchain designs depart from traditional sharded BFT architectures by avoiding explicit network or committee partitioning, and instead rely on deterministic ordering to enable parallel execution~\cite{gelashvili2023block, blackshear2023sui}. As a result, these designs are conceptually closer to the Multi-BFT paradigm. In these systems, transactions that may conflict are assigned a pre-determined global order. By eliminating dynamic coordination during execution, this approach substantially reduces the number of communication rounds required. However, it still incurs overhead for establishing and distributing the global order.

\begin{table}[t]
\centering
\caption{Cross-shard / cross-instance transaction coordination overhead in representative BFT systems. $n$: number of replicas, $k$: number of shards/instances.}
\label{tab:cross-shard-complexity}
\begin{tabular}{l c}
\toprule
\textbf{System} & \textbf{Coordination messages}\\
\midrule
OmniLedger~\cite{omniledger} & $O(k)$  \\
Chainspace~\cite{al2017chainspace} & $O(k^2) $  \\
Pyramid~\cite{hong2021pyramid}  & $O(n)$  \\
\textbf{\sysname} & $none$  \\
\bottomrule
\end{tabular}
\end{table}

\subsection{Coordination Overhead Comparison}
To compare different designs, we consider a transaction that spans $k$ shards/instances and estimate the number of extra coordination messages of different protocols in Table~\ref{tab:cross-shard-complexity}. 

\begin{itemize}
    \item In OmniLedger~\cite{omniledger}, the client collects Proof-of-Acceptance from all input shards and submits the transaction to the output shard. Each involved shard lerader sends at least one message to the client, resulting in a total coordination message complexity of $O(k)$.
    \item During the prepare phase of Chainsapce~\cite{al2017chainspace}, each shard leader needs to validate with the other $k-1$ shards, and then subsequently submits the collected lock proofs to the output shards. In the worst case, the total message complexity is $O(k^2)$, which comes from shard-to-shard verification.
    \item Pyramid~\cite{hong2021pyramid} selects intersection replicas that belong to multiple shards to process cross-shard transactions. Each intersection node participates simultaneously in the BFT consensus of the relevant shards,and then synchronize the cross-shard commit within each shard. This yields $O((n/k) \cdot k) = O(n)$ messages, roughly linear in total system size.
    \item \sysname does not require any cross-instance coordination messages. Cross-instance transactions are executed correctly using locks and deadlock-resolution mechanisms without sending additional messages. Therefore, the coordination messages can be considered as none.
\end{itemize}

\end{document}